\documentclass[a4paper,10pt]{article}

\usepackage{showlabels}

\usepackage{hyperref}

\usepackage{graphicx}
\usepackage{amsmath}
\usepackage{bbm}
\usepackage{deleq}
\usepackage{cite}
\usepackage[sort&compress,numbers]{natbib}
\usepackage{amsfonts}
\usepackage{amssymb}
\usepackage{amsmath}
\usepackage{latexsym}
\usepackage{color}
\usepackage{url}
\usepackage{ntheorem}
\usepackage{braket}
\usepackage[normalem]{ulem}
\usepackage{tikz}
\usepackage{enumerate}

\newcommand{\tvarphiep}{\tilde\varphi_{\epsilon}}

\newcommand{\bzmcD}{\zmcD}

\newcommand{\mrh}{\mathring{h}}

\newcommand{\calB}{{{\mycal B}}}

\newcommand\ben{\begin{enumerate}}
\newcommand\een{\end{enumerate}}
\newcommand\bit{\begin{itemize}}
\newcommand\eit{\end{itemize}}

\newcommand{\lambdatwo}{\lambda_1{}}
\newcommand{\qedskip}{\qed\medskip}

\newcommand{\cerc}{{\mathbb S}}
\newcommand{\ellm}{\red{m}}

\newcommand{\transversemanifold}{N^{n-1}}
\newcommand{\error}{\textit{l.o.t.}}%

\newcommand{\zmcD}{\,\,\mathring{\!\!\mcD}{}}
\newcommand{\ringh}{\,\mathring{\! h}{}}

\newcommand{\red}[1]{{\color{red}#1}}

\newcounter{mnotecount}[section]

\renewcommand{\themnotecount}{\thesection.\arabic{mnotecount}}

\newcommand{\mnote}[1]
{\protect{\stepcounter{mnotecount}}$^{\mbox{\footnotesize
$
\bullet$\themnotecount}}$ \marginpar{
\raggedright\tiny\em
$\!\!\!\!\!\!\,\bullet$\themnotecount: #1} }

\newcommand{\RR}{\R}

\newtheorem{theorem}{\sc  Theorem\rm}[section]

\newtheorem{thm}[theorem]{\sc  Theorem\rm}

\newtheorem{corollary}[theorem]{\sc  Corollary\rm}
\newtheorem{Corollary}[theorem]{\sc  Corollary\rm}

\newtheorem{lemma}[theorem]{\sc Lemma\rm}

\newtheorem{Proposition}[theorem]{\sc Proposition\rm}
\newtheorem{remark}[theorem]{\sc Remark\rm}
\newtheorem{Remark}[theorem]{\sc Remark\rm}

\newcommand{\bbR}{\mathbb{R}}

\newcommand{\jlcax}[1]{}
\newcommand{\eean}{\nonumber\end{eqnarray}}

\newcommand{\kk}[1]{}

\newcommand{\beq}{\begin{equation}}

\newcommand{\FS}       
                  {F}

\newcommand{\HS} 
       {H_{\mbox{\scriptsize volume}}}

{\ptc{this should be removed in the oberwolfach version}}%

\newcommand{\eeal}[1]{\label{#1}\end{eqnarray}}
\newcommand{\bed}{\begin{deqarr}}
\newcommand{\eed}{\end{deqarr}}
\newcommand{\bedl}[1]{\begin{deqarr}\label{#1}}
\newcommand{\eedl}[2]{\arrlabel{#1}\label{#2}\end{deqarr}}

\newcommand{\can}{g_{\cerc^{n-1}}}

\newcommand{\bel}[1]{\begin{equation}\label{#1}}
\newcommand{\bea}{\begin{eqnarray}}
\newcommand{\bean}{\begin{eqnarray}\nonumber}
\newcommand{\beal}[1]{\begin{eqnarray}\label{#1}}
\newcommand{\eea}{\end{eqnarray}}

\newcommand{\nn}{\nonumber}

\newcommand{\qed}{\hfill $\Box$ \medskip}

\newcommand{\be}{\begin{equation}}
\newcommand{\eeq}{\end{equation}}
\newcommand{\ee}{\end{equation}}
\newcommand{\beqa}{\begin{eqnarray}}
\newcommand{\eeqa}{\end{eqnarray}}
\newcommand{\beqan}{\begin{eqnarray*}}
\newcommand{\eeqan}{\end{eqnarray*}}
\newcommand{\ba}{\begin{array}}
\newcommand{\ea}{\end{array}}

\newcommand{\mcD}{{\mycal D}}

\newcommand{\warn}[1]
{\protect{\stepcounter{mnotecount}}$^{\mbox{\footnotesize
$
\bullet$\themnotecount}}$ \marginpar{
\raggedright\tiny\em
$\!\!\!\!\!\!\,\bullet$\themnotecount: {\bf Warning:} #1} }

\newcommand{\R}{\mathbb R}

\newcommand{\N}{\mathbb N}

\newcommand{\eq}[1]{(\ref{#1})}

\newcommand{\Mext}{M_\ext}
\newcommand{\ext}{\mathrm{ext}}

\newcommand{\ptc}[1]{\mnote{{\bf ptc:}#1}}

\newcommand{\beqar}{\begin{deqarr}}
\newcommand{\eeqar}{\end{deqarr}}

\newcommand{\beaa}{\begin{eqnarray*}}
\newcommand{\eeaa}{\end{eqnarray*}}

\newcommand{\tr}{\mbox{tr}}
\newcommand{\zg}{\mathring{g}}

\newcommand{\zGamma}{\mathring{\Gamma}}
\newcommand{\znabla}{\mathring{\nabla}}

\renewcommand{\red}[1]{#1}

\DeclareFontFamily{OT1}{rsfs}{}
\DeclareFontShape{OT1}{rsfs}{m}{n}{ <-7> rsfs5 <7-10> rsfs7 <10-> rsfs10}{}
\DeclareMathAlphabet{\mycal}{OT1}{rsfs}{m}{n}

{\catcode `\@=11 \global\let\AddToReset=\@addtoreset}
\AddToReset{equation}{section}
\renewcommand{\theequation}{\thesection.\arabic{equation}}

{\catcode `\@=11 \global\let\AddToReset=\@addtoreset}
\AddToReset{figure}{section}

{\catcode `\@=11 \global\let\AddToReset=\@addtoreset}
\AddToReset{table}{section}

\def\Tone{{\stackrel{(n)}{T}}}
\def\Ttwo{{\stackrel{(n+2)}{T}}}
\def\Tohalf{{\stackrel{(n+1)}{T}}}
\def\TohalfX{{\stackrel{(*)}{T}}}
\def\OmegatohalfX{{\stackrel{(*)}{\Omega}}}

\begin{document}
\title{On the mass aspect function and positive energy theorems for asymptotically hyperbolic manifolds\protect\thanks{Preprint UWThPh-2017-37}}

\author{Piotr T. Chru\'{s}ciel\thanks{Faculty of Physics and Erwin Schr\"odinger Institute, University of Vienna, and Institut de Hautes \'Etudes Scientifiques, Bures-sur-Yvette} {} \thanks{
{\sc Email} \protect\url{piotr.chrusciel@univie.ac.at}, {\sc URL} \protect\url{homepage.univie.ac.at/piotr.chrusciel}}\\
Gregory J. Galloway\thanks{University of Miami and Erwin Schr\"odinger Institute, Vienna}\\
 Luc
Nguyen\thanks{University of Oxford and Erwin Schr\"odinger Institute, Vienna}\\
Tim-Torben Paetz\thanks{Faculty of Physics and Erwin Schr\"odinger Institute, University of Vienna}}
\maketitle

\begin{abstract}
We prove positivity of energy for a class of asymptotically locally hyperbolic manifolds in dimensions $4\le n \le 7$. The result is established by first proving  deformation-of-mass-aspect theorems in dimensions $n\ge 4$.
Our positivity results extend  to the case $n = 3$ when more stringent conditions are imposed.
\end{abstract}

\tableofcontents
\section{Introduction}

An interesting global invariant of asymptotically hyperbolic manifolds is provided by the total mass, for general conformal boundaries at infinity, or the total energy-momentum vector when the conformal structure at conformal infinity is that of a round sphere~\cite{ChHerzlich,Wang,ChNagyATMP} (compare~\cite{AbbottDeser,ChruscielSimon}).  These objects provide a generalisation of the  Arnowitt-Deser-Misner (ADM) energy-momentum, which is defined for asymptotically flat manifolds, to the asymptotically hyperbolic case. While there are by now sharp positivity results for the ADM mass in all dimensions~\cite{SchoenYau2017}, the asymptotically hyperbolic case is still poorly understood. The purpose of this work is to expand somewhat our understanding of the topic.

As such, our main result is (see Section~\ref{s18VIII15.1} below for  notation
{and terminology}):

\begin{thm}\label{pmass1}
 Let $(M^n,g)$, $4 \le n \le 7$, be a
 $C^{n+5}$--conformally compactifiable
asymptotically locally hyperbolic (ALH) Riemannian manifold
diffeomorphic to  $[r_0,\infty) \times \transversemanifold$ with a compact boundary $N_0 := \{r_0\} \times \transversemanifold$ and with well defined total mass. Suppose that:

\begin{enumerate}
\item The mean curvature of $N_0$ satisfies $H< n-1$,  where $H$ is the    divergence $D_i n^i$ of the unit normal $n^i$ pointing into $M$.
\item The scalar curvature $R = R[g]$ of $M$ satisfies $R \ge -n(n-1)$.
\item Either $(N ,\mathring{h})$ is a  flat  torus, or  $(N ,\mrh)$ is a nontrivial quotient of a round sphere.
\een

Then the mass of $(M^n,g)$ is nonnegative, $m \ge 0$.
\end{thm}

\begin{remark}
  \label{R30XII17.11}
{\rm
It should be  clear from its proof below that Theorem \ref{pmass1} remains valid in the case $n =3$ if one assumes in addition that the mass aspect function has a sign.
}
\end{remark}

Note that the above applies in particular to manifolds with a minimal boundary  $H=0$, which arise in general relativity in time-symmetric initial data sets with apparent horizons.

It might be worthwhile pointing out that the assumed product structure on $M$  arises in certain technical aspects of the proof.  For example, in the torus case, the proof requires the existence of a deformation retract of the conformal compactification of $M$ onto its conformal boundary.  The product structure assumption is the simplest condition to ensure the existence of this, although somewhat more general topologies could be allowed.
In the spherical space case, the product structure, which in fact we assume extends to the conformal completion, is used to control the structure of the universal cover.
The well-known examples of \cite{Birmingham} have product topology.

We do not address   the question of rigidity in the case $m = 0$.  Our proof involves an initial perturbation of the metric (using Theorem \ref{T31VII17.1a} below) to a metric which may not have vanishing mass, and as such, may not have vanishing mass aspect.  Hence, for example, the analysis of the sort given in \cite[Section 3.2]{AnderssonGallowayCai} does not seem to be of use in our context.

Now, the total energy, or energy-momentum, are defined by integrating a function, called the \emph{mass aspect}, over the conformal boundary.
Part of the proof of Theorem~\ref{pmass1} consists in an analysis of this function, which has some interest of its own.
Here some terminology is required:
we will say that a function $f$ on $\mathbb{S}^{n-1}=\{y\in \R^n\,,\ |y|=1\}$ is a \emph{monopole-dipole function} if $f$ is a linear combination of constants and the functions $\theta^i= y^i/|y|$. We have:

\begin{theorem}
  \label{T31VII17.1a}
Let $(M^n,g)$ be an
  ALH manifold, $n\ge 4$ with
 $C^k$--conformal compactification, $k\ge 3$,
and with well-defined mass aspect function.
For all $ \epsilon>0 $ there exists a metric $g_\epsilon$ which is  $C^{\min(k,n+1)}$--conformally compactifiable when $n=4$, and $C^k$--conformally compactifiable otherwise,
which coincides with $g$ outside of an $\epsilon$-neighborhood of the conformal boundary at infinity, satisfies $R[g_\epsilon]\ge R[g ] $, and which has a well-defined mass aspect function
such that
\begin{enumerate}

 \item $g_\epsilon$ has a pure monopole-dipole mass aspect function
   $\Theta_\epsilon$ if $(\transversemanifold, \ringh)$ is conformal to the standard sphere,
   and has constant mass aspect function otherwise;

  \item the associated energy-momentum satisfies
  \begin{equation}\label{10VII17.2a}
   \left\{
     \begin{array}{ll}
       \lim_{\epsilon \to 0} m^\epsilon_0 = m_0\,,
\
 m^\epsilon_i = m_i \,,
& \hbox{if $(\transversemanifold, \ringh)$ is conformal to the round $ {\mathbb S}^{n-1}$;} \\
       \lim_{\epsilon \to 0} m^\epsilon = m
 \,, & \hbox{otherwise.}
     \end{array}
   \right.
\ee
\end{enumerate}
\end{theorem}

See Remark~\ref{R4V18.1} below for more information on the differentiability of the metrics $g_\epsilon$ when $n=4$.

Theorem~\ref{T31VII17.1a}  has the following corollary:

\begin{Corollary}
 \label{C10VIII17.1}
  Under the conditions above, suppose that $(\transversemanifold, \ringh)$ is conformal to the standard sphere and that the energy-momentum covector $(m_0, m_1, \ldots, m_n)$ defined by \eqref{9VIII17.Add1} below is timelike: $m_0^2 - \sum_{i \geq 1} m_i^2 > 0$. Then there exists a metric $g_\epsilon$ as in Theorem~\ref{T31VII17.1a}  which has a constant mass aspect function in a suitable conformal frame at infinity.
\end{Corollary}

Corollary~\ref{C10VIII17.1}  and a more precise version of Theorem~\ref{T31VII17.1a} are proved in Section~\ref{s17VIII15.3} below. Further deformation results can also be found there.

The restriction $n\ge 4$ is necessary in our analysis of the mass aspect function. This is due to the fact that our deformation procedure introduces error terms with a dimension-dependent decay rate. The method we use to compensate these error terms turns out to work if $n\ge 4$, but we have not been able to devise a technique to absorb the errors when $n=3$. On the other hand, the restriction $n\le 7$ in Theorem~\ref{pmass1} arises from the regularity theory of CMC hypersurfaces. It is conceivable that a generalisation of the methods of Schoen and Yau~\cite{SchoenYau2017} to asymptotically hyperbolic manifolds will allow one to remove the upper bound on $n$ in the positivity results here.

The fact that the mass aspect cannot be deformed to a constant in the spherical case is not surprising. Indeed, when conformal infinity is spherical the total energy is not a number but a vector, and the first non-trivial spherical harmonics of the mass aspect determine its spatial components.  In particular a constant mass aspect implies timelikeness of the energy-momentum vector. Our deformation procedure is devised to change the total energy-momentum by an arbitrarily small amount, and a deformation procedure which would change the causal character of the total energy-momentum is incompatible with the small-change requirement.

To put our studies of the mass aspect function in a wider context, recall that  Lee and Neves established a Penrose-type inequality for a class of three-dimensional asymptotically hyperbolic manifolds~\cite{LeeNeves} under the assumption that the mass aspect function has constant sign.  A similar hypothesis has been made previously by Andersson, Cai and Galloway in their proof of positivity of hyperbolic mass, in dimensions $3\le n\le 7$ and without the hypothesis that the manifold is spin~\cite{AnderssonGallowayCai}. The results in~\cite{ChDelayAH} imply that the hypothesis of constant sign of the mass aspect function can be removed under smallness assumptions, or with a fast dimension-dependent decay rate of the metric towards model solutions.  However, one would like to remove such supplementary assumptions altogether. We have unfortunately not been able to achieve this, in particular the restriction on dimension $n\ge 4$ renders our result useless for improving the Lee-Neves theorem. On the other hand, Theorem~\ref{T31VII17.1} provides the following minor improvement of the Andersson-Cai-Galloway theorem, keeping in mind that their hypothesis of mass aspect of constant sign implies that the energy-momentum vector is timelike (see Section~\ref{s18VIII15.1} for terminology):

\begin{theorem} \label{thm:posmass}
Let $(M^{n},g)$, $4 \leq n \leq 7$, be a
manifold with scalar curvature $R[g] \ge -n(n-1)$ with a metric which is smoothly conformally compactifiable with spherical conformal infinity.
Suppose that \eqref{Rabc5+} below holds with $\beta=n$ and assume that   $R[g]+n(n-1)=O(x^{n+1})$.
 Then the total energy-momentum vector of $(M^{n},g)$ \emph{cannot} be timelike past-pointing.
\end{theorem}

It is clear that the asymptotic hypotheses in Theorem~\ref{thm:posmass} can be weakened, but this is irrelevant for our purposes here.

We note that examples of metrics with constant negative scalar curvature and with a null or  spacelike energy-momentum vector on a (non-complete) asymptotically hyperbolic manifold have been constructed by Cortier in~\cite{CortierMass}.

 We stress that our analysis concerns the mass of asymptotically hyperbolic metrics, which coincides with the standard definitions of total mass of asymptotically anti-de Sitter spacetimes {only} when the usual no-radiation conditions at the timelike conformal boundary at infinity are imposed. In particular we do not cover those asymptotically hyperboloidal initial data sets with $\Lambda=0$ which intersect a null conformal boundary at infinity at a cut on which the radiation field does not vanish, nor initial data sets in asymptotically anti-de Sitter spacetimes which meet the conformal boundary at infinity in an unusual manner.

\section{The hyperbolic mass}
 \label{s18VIII15.1}

We briefly review part of~\cite{ChHerzlich}  as relevant for our purposes here.

Consider a manifold $M$ with a metric $g$ which asymptotes to a \emph{reference metric ${\zg}$}, and
contains a region $\Mext\subset M$ of the form
\be\label{Nman}
 \Mext = [r_0,\infty)\times \transversemanifold\,,\ee
where $\transversemanifold$ is a compact $(n-1)$-dimensional boundaryless manifold, $n\ge 3$, such that the {reference
metric ${{\zg}}$} on $\Mext$ reads
\be
\label{cm1}
{\zg}:=  \frac{dr^2}{r^2+k}+r^2{{\ringh}{}}\,, \ee with
${{\ringh}{}}$
being  a Riemannian metric on $\transversemanifold$ with scalar  curvature $R[{{{\ringh}{}}}]$  equal to
\be\label{cm1.1}
R[{{{\ringh}{}}}]= (n-1)(n-2)k\,, \quad k\in\{0,\pm 1\}
\,.
\ee
Here and below $r$ is a coordinate running along the $[r_0,\infty)$
factor of $[r_0,\infty)\times \transversemanifold$.

As an example, the metric on the time slices in the Schwarzschild - anti de Sitter (Kottler) spacetime is (compare~\cite{Birmingham})
\begin{equation}\label{20IV15.1}
 g_m = \frac{dr^2}{\frac {r^2}{\ell^2} + k - \frac {2m } {r^{n-2}}} + r^2 {\ringh}{}
 \,,
\ee
which asymptotes to \eq{cm1} as $r\to\infty$ after a constant rescaling of the coordinate $r$ and of the metric.

When $(\transversemanifold,{{\ringh}{}})$  is the unit round $(n-1)$--dimensional sphere
$(\cerc^{n-1},\can)$, then ${{\zg}}$ is the hyperbolic metric.

Equations~\eq{cm1} and \eq{cm1.1} imply that the scalar curvature
$R[{{{\zg}}}]$ of the metric ${{\zg}}$ is constant:
 $$ R[{{{\zg}}}]=-
n(n-1)\,.$$

In what follows we will assume that $\zg$ is Einstein. This will be the case  if and only
if ${{\ringh}{}}$ is.
We note that, for the purpose of definition of the mass,  the background metric
${{\zg}}$  needs to be defined only on $\Mext$.

The definition of mass integrals requires appropriate
boundary conditions, which are most conveniently defined using the following
${\zg}$-orthonormal frame $\{f_i\}_{i=1,n}$ on $\Mext$:
\be\label{Rabc1}  f_A = r^{-1}\epsilon_A\,, \quad
A=2,\ldots,n \,,\quad
 f_1 = \sqrt{r^2+k}
\;\partial_r\,,
\ee
where the $\epsilon_i$'s form an orthonormal
frame for the metric ${{\ringh}{}}$. We   set
\be \label{Rabc2}
g_{ij}:=g(f_i,f_j)
\,,
\quad
 e_{ij}:=g_{ij}-\zg_{ij}
 \,.\ee
The coordinate-independence of the mass integrals requires the fall-off conditions
\be
  \label{Rabc5}\sum_{i,j}
 |g_{ij}-\delta_{ij}| + \sum_{i, j,k} |f_k(g_{ij})|=
 o(r^{-n/2})\,.
\ee

Recall that \emph{static Killing Initial Data (KIDs)} are defined as the set of solutions of the equations
\begin{equation}\label{9IV18.31}
  \znabla_i \znabla_j V = V (R[\mathring{g}]_{ij} -\lambda \mathring{g}_{ij})
\end{equation}
where $\lambda$ is related to the cosmological constant $\Lambda$ as $\lambda = - n \ell^{-2}$, with $\ell^2
 =
 -
 \frac{n(n-1)}{2\Lambda}$. (In most of this work the constant $\ell$ will be scaled away to $1$, which together with the assumption that $\mathring{g}$ is Einstein yields $\znabla_i \znabla_j V = V\,\mathring{g}_{ij}$.)

Ignoring momentarily issues associated with the dimension of the space of static KIDs (to be addressed shortly), when $\ell$ is scaled to $1$ the mass is defined as
 \begin{eqnarray}
 \nn
  m
&=&\lim_{R\to\infty}  (R^2+k)\times \frac 1 {16 \pi} \times
\\
 &&
   \displaystyle \int_{\{r=R\}} \left(-\sum_{A=2}^{n }\left\{\frac {
          \partial e_{AA}}{\partial r}+ \frac {
          k e_{AA}}{ r(r^2+k)}\right\}+\frac {(n-1)e_{11}}{r}
   \right) d \mu_{h}\,,
   \label{massequation1}
 \end{eqnarray}
where $d \mu_{h}$ is the Riemannian measure associated with
the metric $h$ induced by $g$ on the level sets of the function $r$. The existence of the limit
is guaranteed by the conditions
\begin{deqarr}
 \arrlabel{Rabc3}
& \int_{\Mext} \left( \sum_{i,j} |g_{ij}-\delta_{ij}|^2 + \sum_{i,
j,\ell} |f_\ell(g_{ij})|^2 \right)r \;d\mu_g<\infty\,,
 \label{Rabc3a}&
\\
 & \int_{\Mext}
 |R[{{ g}}]-R[{{{\zg}}}] |\;r \;d\mu_g<\infty\,,\label{Rabc3b}&
\end{deqarr}
\begin{equation}
\label{Rabc0} \exists \ C > 0 \ \textrm{ such that }\
C^{-1}\zg(X,X)\le g(X,X)\le C\zg(X,X)\,.
\end{equation}

Let $\beta>0$. We will say that a metric is $\beta$-asymptotically hyperbolic if
\be
  \label{Rabc5+}
   \sum_{i,j}
 |g_{ij}-\delta_{ij}| + \sum_{i, j,k} |f_k(g_{ij})|=
 O(r^{-\beta})\,.
\ee
We note that  both \eq{Rabc5} and (\ref{Rabc3a})
 will hold if $\beta>n/2$.

The above has a natural formulation in terms of manifolds $M$ with  boundary $\partial M$, where one or more connected components of $\partial M$ are viewed as a conformal boundary at infinity.   In the setup above, the conformal boundary at infinity is diffeomorphic to $\transversemanifold$. For simplicity we will assume that $\partial M$ has only one component, which is a boundary at infinity, as the generalisations are straightforward. In this context let $x$ be a smooth function defined on $M$ which vanishes precisely on those components of $\partial M$, with $\mathrm{d}x$ nowhere vanishing on $\partial M$. A metric $g$ on $M$ is said to be smoothly, respectively $C^{k,\alpha}$, conformally compactifiable if the metric $x^2 g$ extends smoothly, respectively $C^{k,\alpha}$, across $\partial M$. 

Relevant for this work is a class of conformally compactifiable metrics which can be written as
\begin{eqnarray}
 \label{18VIII15.6}
  g &=& \ell^2 x^{-2} \Big(\mathrm{d}x^2  + (1-\frac{k}{4}x^2)^2 {\ringh}{}     + x^{n} \mu \Big)+o(x^{n-2}) \mathrm{d}x^i \mathrm{d}x^j
\,,
\\
 \label{18VIII15.7}
{\ringh}{}&=&\mathring  h_{AB}(x^C)\mathrm{d}x^A\mathrm{d}x^B\,,
\\
 \label{18VIII15.8}
\mu &=&\mu_{AB}(x^C)\mathrm{d}x^A\mathrm{d}x^B
\,,
\end{eqnarray}
where $\ell>0$ is a constant, where the $x^A$'s, $A=2,\ldots,n$,
are local coordinates on  $\transversemanifold$, and where $(x^i)=(x,x^A)$.
Here, as elsewhere, expressions such as $o(x^{p})\mathrm{d}x^A\mathrm{d}x^B$ mean  $f_{AB}\mathrm{d}x^A\mathrm{d}x^B$ with
$f_{AB}=o(x^{p})$; $O(x^{p})\mathrm{d}x^A\mathrm{d}x^B$ are similarly defined. Under suitable further differentiability conditions, such metrics are referred to as \emph{asymptotically locally hyperbolic} in~\cite{LeeNeves}. They are called \emph{asymptotically hyperbolic} in~\cite{AnderssonGallowayCai} when in addition one assumes that $\transversemanifold$ is diffeomorphic to $\cerc^{n-1}$ with the unit round metric.

Suppose that $\ell=1$, which can be achieved by a constant rescaling of $g$. Replacing $x$ by a coordinate $r$ through the formula
$$
 \frac{\mathrm{d}x}x =  -\frac{dr}{  \sqrt{ {r^2}  + k }
 }
 \,,
$$
and observing that the metric ${\zg}$ defined in \eqref{cm1} is transformed to
\begin{eqnarray}
  {\zg} &=& x^{-2} \Big(\mathrm{d}x^2
   + (1-\frac{k}{4}x^2)^2 {\ringh}{} \Big)
\,,
 \label{17VIII15.1}
\end{eqnarray}
one can bring \eq{18VIII15.6} to the form needed for the definition of mass. For such metrics, \eq{massequation1} can be rewritten as
\begin{equation}\label{23VIII15.1}
 m = c_n \int_{\transversemanifold} \tr_{{\ringh}{}}\mu\, d\mu_{{\ringh}{}}
 \,,
\ee
where $c_n$ is some universal normalising positive  constant depending only on $n$, and where  the integrand
\begin{equation}
 \Theta \equiv \tr_{{\ringh}{}}\mu := {\ringh}{}^{AB} \mu_{AB}
 	\label{28VII17.MAs}
\end{equation}
is called the \emph{mass aspect function}.

When $(\transversemanifold,\ringh)$ is \emph{not} the standard sphere $(\mathbb{S}^{n-1},h_0)$, \eqref{23VIII15.1} defines a geometric invariant of $g$: it is independent of the choice of coordinate systems in which the asymptotics \eqref{18VIII15.6}-\eqref{18VIII15.8} holds.
On the other hand, when $(\transversemanifold,\ringh)$ is the standard sphere $(\mathbb{S}^{n-1},h_0)$, the number $m$ defined in \eqref{23VIII15.1} is coordinate dependent. While it is invariant under coordinate transformations which pointwise fix the boundary at infinity, there are asymptotic coordinate transformations which preserve \eqref{18VIII15.6}-\eqref{18VIII15.8} but not \eqref{23VIII15.1}. In this case one considers instead  the \emph{energy-momentum covector} $(m_0, m_1, \ldots, m_n)$ defined by
\begin{equation}
m_0 = c_n \int_{\mathbb{S}^{n-1}} \tr_{{\ringh}{}}\mu \,d\mu_{{\ringh}{}} \text{ and } m_i = c_n \int_{\mathbb{S}^{n-1}} \tr_{{\ringh}{}}\mu\,X_i \,d\mu_{{\ringh}{}}
 \,,
	\label{9VIII17.Add1}
\end{equation}
where $X_1, \ldots, X_n$ are normalized first eigenfunctions on $\mathbb{S}^{n-1}$ which form an orthogonal basis of the first eigenspace of the Laplacian on $\mathbb{S}^{n-1}$: $\int_{\mathbb{S}^{n-1}} X_i\,X_j\,d\mu_{\ringh} = \frac{1}{n} \textrm{Volume}(\mathbb{S}^{n-1})\,\delta_{ij}$. The number
\[
m_0^2 - \sum_{i \geq 1} m_i^2
\]
and the causal character of the energy-momentum covector $(m_0, m_1, \ldots, m_n)$ are then geometric invariants of $g$; see Section~\ref{ssec:HypSym} below, compare~\cite{ChHerzlich,ChruscielSimon,ChNagy,Wang,ChNagyATMP}.

\section{Changing the mass aspect function}
 \label{s18VIII15.2}

We wish to analyse how the mass aspect function behaves under a certain class of coordinate transformations. To this end, it is necessary to consider metrics more general than those of the form {\eqref{18VIII15.6}-\eqref{18VIII15.8}}, as such form is not preserved under the coordinate transformations we would like to perform.

\subsection{Perturbations of infinity at order $x^n$}

Consider now metrics which are similar to {\eqref{18VIII15.6}-\eqref{18VIII15.8}} but allow for $\mathrm{d}x^2$- and $\mathrm{d}x\,\mathrm{d}x^A$- terms:
\begin{equation}
g
	 =  \mathring{g} + \underbrace{\lambda_{xx}\,\mathrm{d}x^2 + 2\,\lambda_{xA}\,\mathrm{d}x\,\mathrm{d}x^A + \lambda_{AB}\,\mathrm{d}x^A\,\mathrm{d}x^B}_{=:\lambda}
	\,,
 \label{28VII17.1}
\end{equation}
where
\begin{equation}
\lambda_{xx} = O_2(x^{n-2}),\ \lambda_{AB} = O_2(x^{n-2})\ \text{ and }\ \lambda_{xA} = O_2(x^{n-3})
	\,.
 \label{28VII17.2}
\end{equation}
Here and below,
we write
\begin{equation}\label{18XII17.22}
 f = O_\ell(\Psi(x))
\end{equation}
for some positive function $\Psi$ if
for $0\le i \le \ell$ we have  $|\mathring \nabla^i  f|_{\zg}  \leq C\Psi(x)$ for small $x$, for some constant $C$.

Instead of making the necessary changes of variables to bring $g$ to the form needed to evaluate \eqref{massequation1}, one can read off the mass integral directly as follows. By \eqref{26VIII15.1asf}, Appendix~\ref{app17VIII15.1} below, the scalar curvature $R[g]$ of $g$ satisfies
\begin{align}
R[g] - R[\mathring{g}]
	&= -x^{n+1}\partial_x\Big\{x^{-(n-1)} \partial_x (x^2 \ringh^{AB}\,\lambda_{AB})\nonumber\\
		&\qquad
		- 2x^{-(n-3)} \bzmcD ^A \lambda_{xA}
		+ (n-1)x^{-(n-2)}\lambda_{xx}\Big\}
 \nn
\\
		&\qquad  + x^4\zmcD^A \zmcD^B \lambda_{AB} + O(x^{n+2})
		\,.\label{28VII17.3}
\end{align}
The mass/energy-momentum covector is recognized as the flux integral(s) related to the above expression against suitable KID potential(s) (i.e. the functions $V$ in \eqref{9IV18.31}). When $(N,\ringh)$ is not conformal to the round sphere, $V$ is taken to be $V = x^{-1}(1 + \frac{k}{4} x^2)$. When $(N,\ringh)$ is conformal to the round sphere $\mathbb{S}^{n-1}$, $V$ can be taken to be $V_0 = x^{-1}(1 + \frac{1}{4}x^2)$ and $V_i = x^{-1}(1 - \frac{1}{4}x^2) X_i$ where $X_1, \ldots, X_n$ are normalized first eigenfunctions on $\mathbb{S}^{n-1}$ as in Section \ref{s18VIII15.1}.
 In particular, if we assume that
$x^{-1} (R[g] - R[\mathring{g}])\in L^1$ and
\begin{align}
\lambda
	&= x^{n-2}\Big[\mu_{xx}(x^C)\, \mathrm{d}x^2 + \frac{2}{x}\mu_{xA}(x^C)\,\mathrm{d}x\mathrm{d}x^A + \mu_{AB}(x^C)\,\mathrm{d}x^A\mathrm{d}x^B\Big]\nonumber\\
		&\qquad + o_2(x^{n-2})\mathrm{d}x^2 + o_2(x^{n-3}) \mathrm{d}x\,\mathrm{d}x^A + o_2(x^{n-2})\mathrm{d}x^A\,\mathrm{d}x^B
			\,,\label{28VII17.4}
\end{align}
then the mass of $g$ is found to be
\begin{equation}
m = c_n\int_{\transversemanifold}\Big[ \ringh^{AB} \mu_{AB} - \frac{2}{n} \bzmcD ^A\,\mu_{xA}
  + \frac{n-1}{n} \mu_{xx} \Big]\, d\mu_{{\ringh}{}}
	\,,\label{28VII17.5}
\end{equation}
when $(N,\ringh)$ is not conformal to the round sphere. In the other case, the energy-momentum covector of $g$ can be similarly computed by integrating the expression in the square bracket in \eqref{28VII17.5} against the constant one and the functions $X_i$.

The integrand on the right-hand side of \eqref{28VII17.5} contains a divergence term which does not contribute to the  {integral}, but we keep it in this form as it coincides with $\frac{1}{n}$ times the leading term of the sum contained in the curly brackets on the right-hand side of \eqref{28VII17.3}. We will continue to refer to this quantity as the \emph{mass aspect function} (for metrics given by \eqref{28VII17.1} and \eqref{28VII17.4}):
\begin{equation}
\Theta:= \ringh^{AB} \mu_{AB} - \frac{2}{n} \bzmcD ^A\,\mu_{xA} + \frac{n-1}{n} \mu_{xx}
	\,.\label{28VII17.6}
\end{equation}

\begin{lemma}\label{Lem:MAInv}
The mass aspect function is invariant under coordinate transformations which pointwise preserve infinity and the asymptotic  behavior in \eqref{28VII17.4}. Equivalently, under a transformation of the form
\begin{equation}\label{29VII17.1}
(x,x^A) \mapsto \big(y
  = x + x^{n+1}\,\psi(x^C) + o_{3}(x^{n+1})\,,\
    y^A = x^A + x^n\,X^A(x^C) + o_{3}(x^n)\big)
 \,,
\end{equation}
where $\psi$ and $X^A$ are of $C^3(\transversemanifold)$-differentiability class,
the new mass aspect function $\tilde \Theta$ and the original mass aspect function $\Theta$ satisfy
\[
\tilde \Theta(y^C) = \Theta(y^C).
\]
\end{lemma}

\begin{remark}
 \label{R29VII17.1}
 {\em
 If we assume that, in local coordinates, the metrics  are $C^k$-conformally compactifiable both before and after the coordinate transformation,
 with $k$ large enough,
 then  \eqref{29VII17.1} exhausts the set of transformations described in the first sentence of the Lemma.  This follows essentially from \cite[Equations~(3.18)-(3.20)]{ChNagyATMP}: Indeed, it is standard to go from the estimates there to the expansions \eqref{29VII17.1} with a loss of derivatives. A conservative estimate is $k\ge n+7$, and it is clear that a careful argument can bring the threshold down. Compare Proposition~\ref{Prem:02V18-R1} below, where   supplementary conormal regularity is imposed.
 }
\end{remark}

\begin{remark}
  \label{R27VIII17.1}
{\em
We will consider various coordinate transformations such as  \eq{29VII17.1}, which a priori only make sense in local charts. To make global sense of such formulae, in particular to see that the coefficients $X^A$ naturally define a vector field on $\transversemanifold$, one can proceed as follows:
Let $2\le \ellm \in \N$
%
 and consider a metric $g$ of the form
\[
g = \mathring{g} + x^{\ellm -2} \mu_{AB}\,\mathrm{d}x^A\,\mathrm{d}x^B +o(x^{\ellm -2}) \mathrm{d}x^i \mathrm{d}x^j
	\,.
\]

Let $X=X^A\partial_A $ be a vector field on $\transversemanifold$,
 and let $\zeta$
  be a parameter along the flow of $X$. Thus
\begin{eqnarray}
 \nn
  \lefteqn{
 \frac{\mathrm{d}x^A(\zeta)}{\mathrm{d}\zeta} =X^A(x^B(\zeta))\,, \quad x^A(0)=x^A
  }
 &&
\\
 &&  \Longleftrightarrow \quad
x^A (\zeta)=\phi^A[X](\zeta,x^B)\,, \quad \phi^A[X](0,x^B)=x^A
\,.
\end{eqnarray}
One can then pass to a new coordinate system $(x,x^A)\mapsto (x,\phi^A) $ by setting $\zeta= x^\ellm $ in the flow:
\begin{eqnarray*}
 \phi^A:= \phi^A[X](x^\ellm,x^B) =  x^A +X^A  x^\ellm  + O(x^{2\ellm })
\,,
\end{eqnarray*}
which is essentially \eq{29VII17.1}. The transformation formulae for the expansion coefficients of the metric, in terms of powers of $x$ near $x=0$, are then obtained by the usual calculations involving flows.
%
}
\end{remark}

\noindent{\sc Proof of Lemma~\ref{Lem:MAInv}:} We compute%
\begin{align*}
x^{-2}\Big(1 - \frac{k}{4}x^2\Big)^2
	&= y^{-2}\Big(1 - \frac{k}{4}y^2\Big)^2 + 2y^{n-2}\,\psi(y^C) +o(y^{n-2})
		\,,\\
\mathrm{d}x
	&= [1 - (n+1)y^{n}\psi(y^C) +o(y^{n})
 ]\mathrm{d}y\\
		&\qquad\qquad - [y^{n+1}\bzmcD _A  \psi(y^C)+o(y^{n+1})
  ]\,\mathrm{d}y^A
		\,,\\
\mathrm{d}x^A
	&= [-n\,y^{n-1}\,X^A(y^C) +o(y^{n-1})
  ]\mathrm{d}y\\
		&\qquad\qquad + [\delta^A{}_B - y^n\,\partial_B X^A(y^C)
  +o(y^{n})
]\mathrm{d}y^B
		\,,\\
h_{AB}(x^C)
	&= h_{AB}(y^C) - y^n \partial_{D} h_{AB}(y^C)\,X^D(y^C) +o(y^{n})
	\,.
\end{align*}
This implies that
\begin{align*}
x^{-2}\,\mathrm{d}x^2
	&= y^{-2}\,\mathrm{d}y^2  - [2ny^{n-2}\,\psi(y^C)+o(y^{n-2})
  ]\,\mathrm{d}y^2\\
		&\qquad\qquad  + O(y^{n-1})\mathrm{d}y\,\mathrm{d}y^A + O(y^{2n})\mathrm{d}y^A\,\mathrm{d}y^B
			\,,
\end{align*}
and
\begin{align*}
&x^{-2}\Big(1 - \frac{k}{4}x^2\Big)^2\,\ringh_{AB}(x^C)\,\mathrm{d}x^A\,\mathrm{d}x^B\\
	&\qquad = y^{-2}\Big(1 - \frac{k}{4}y^2\Big)^2\,\ringh_{AB}(y^C)\,\mathrm{d}y^A\,\mathrm{d}y^B \\
		&\qquad\qquad - 2ny^{n-3}\,\Big(1 - \frac{k}{4}y^2\Big)^2\ringh_{AB}(y^C)\,X^B(y^C)\,\mathrm{d}y\,\mathrm{d}y^A\\
		&\qquad\qquad + 2y^{n-2}\psi(y^C)\,\ringh_{AB}(y^C)\,\mathrm{d}y^A\,\mathrm{d}y^B\\
		&\qquad\qquad - y^{n-2} \Big(1 - \frac{k}{4}y^2\Big)^2\,\underbrace{[\partial_D \ringh_{AB}\,X^D + 2\ringh_{D(A}\partial_{B)} X^D]}_{=\bzmcD _B X_A + \bzmcD _A X_B}(y^C)\,dy^A\,dy^B
\\
&\qquad\qquad
			+ O(y^{2n-4})\mathrm{d}y^2  +o(y^{n-3}) \mathrm{d}y\,\mathrm{d}y^A + o(y^{n-2})\mathrm{d}y^A\,\mathrm{d}y^B,
\end{align*}
where $X_A = \ringh_{AB}X^B$.

It follows that, in the new coordinate system, the difference $\tilde \lambda$ of $g$ and the new reference metric $\mathring{\tilde g} = y^{-2}(dy^2 + (1 - \frac{k}{4}y^2)^2\,\ringh_{AB}(y^C)\,dy^A\,dy^B$ takes the form
\begin{align}
\tilde \lambda
	&= y^{n-2}\Big[(\mu_{xx}(y^C) - 2n\psi(y^C))\, \mathrm{d}y^2 + \frac{2}{y}(\mu_{xA}(y^C) - nX_A(y^C)) \,\mathrm{d}y\mathrm{d}y^A\nonumber\\
		&\qquad + (\mu_{AB}(x^C) + 2\psi(y^C)\,\ringh_{AB}(y^C) -  (\bzmcD _B X_A + \bzmcD _A X_B)(y^C))\,\mathrm{d}y^A\mathrm{d}y^B\Big]\nonumber\\
		&\qquad + o_2(y^{n-2})\mathrm{d}y^2 + o_2(y^{n-3}) \mathrm{d}y\,\mathrm{d}y^A + o_2(x^{n-2})\mathrm{d}y^A\,\mathrm{d}y^B
			\,.
\label{Eq:29XII17-L1}
\end{align}
We see that
\begin{align*}
\tilde \Theta
	&= (\ringh^{AB}\mu_{AB} + 2(n-1)\psi - 2\bzmcD ^A X_A) - \frac{2}{n}\bzmcD ^A(\mu_{xA} - nX^A)\\
		&\qquad\qquad + \frac{n-1}{n} (\mu_{xx}  - 2n\psi)\\
	&= \ringh^{AB}\mu_{AB} - \frac{2}{n}\bzmcD ^A\mu_{xA} + \frac{n-1}{n} \mu_{xx}\\
	&=\Theta,
\end{align*}
as desired.
\qedskip

As a corollary of the proof, we have
\begin{corollary}\label{Cor:29XII17-C1}
Any metric of the form \eqref{28VII17.4} such that $\lambda_{xA} = o_2(x^{n-2})$ can be put in the form  \eqref{18VIII15.6}-\eqref{18VIII15.8} via a change of coordinates at infinity.
\end{corollary}

\begin{proof} We make a coordinate transformation of the form $(x,x^A) \mapsto \big(y
  = x + \frac{1}{2n}x^{n+1}\,\mu_{xx}(x^C)\,,
    y^A = x^A\big)$. By inspecting the argument leading to \eqref{Eq:29XII17-L1} and using the hypothesis $\lambda_{xA} = o_2(x^{n-2})$, it is readily seen that the $o_2(y^{n-3}) \mathrm{d}y\,\mathrm{d}y^A$ term in \eqref{Eq:29XII17-L1} is in fact $o_2(y^{n-2}) \mathrm{d}y\,\mathrm{d}y^A$. We thus obtain
\begin{align*}
\tilde \lambda
	&= y^{n-2}\Big[\mu_{AB}(x^C) - \frac{1}{n} \mu_{xx}(y^C)\,\ringh_{AB}(y^C) \Big]\,\mathrm{d}y^A\mathrm{d}y^B\\
		&\qquad + o_2(y^{n-2})\mathrm{d}y^2 + o_2(y^{n-3}) \mathrm{d}y\,\mathrm{d}y^A + o_2(x^{n-2})\mathrm{d}y^A\,\mathrm{d}y^B
			\,,
\end{align*}
which completes the proof.
\end{proof}
\qedskip

Let $m\in\N$. We will say that a function $f$ on $\overline M$ is of differentiability class $C^{\ell|m}$ if for any vector fields $X_i$ which are smooth on the compactified manifold and tangent to its boundary  it holds that
\begin{equation}\label{5V18.7}
 \mbox{$\forall \ 0 \le  i\le m$ we have $X_1 \cdots X_i (f)\in C^\ell(\overline M)$.}
\end{equation}
Here the index $i$ does \emph{not} indicate a component of the vector, but numbers the vectors. This definition generalises in the following obvious way to tensor fields $u$: if $\mathring D$ is any smooth covariant derivative operator on $\overline M$, then \eqref{5V18.7} is replaced by
\begin{equation}\label{5V18.8}
 \mbox{$\forall \ 0 \le  i\le m$ we have $\mathring D_{X_1} \cdots \mathring D_{X_i}  u \in C^\ell(\overline M)$.}
\end{equation}

In what follows we will need the following:

\begin{Proposition}\label{Prem:02V18-R1}
Suppose that $g$ is a
$C^{\ell|m}$--conformally compactifiable metric of the form \eqref{18VIII15.6}-\eqref{18VIII15.8}, $m\ge 2$. Then, after a suitable change of coordinates at infinity, in which the metric becomes $C^{\ell-1|m-2}$--conformally compactifiable,
the terms $o(x^{n-2})dx^i dx^j$ in \eqref{18VIII15.6} can be arranged to assume the form $o(x^{n-2})dx^A\,dx^B$.
If $\ell \ge n+1$ and $m\ge 4$  the mass aspect function remains unchanged.
\end{Proposition}

\begin{proof}
We solve
\begin{equation}
|d\ln y(x,x^A)|_g = 1
	\label{Eq:01V18-E1}
\end{equation}
under the boundary condition that $y = 0$ when $x = 0$. Writing
\begin{equation}\label{5V18.1}
y = x\,\exp\chi(x,x^A) \,,
\end{equation}
equation \eqref{Eq:01V18-E1} becomes
\[
1 = g^{xx}(\frac{1}{x} + \partial_x \chi)^2 + 2g^{xA}(\frac{1}{x} + \partial_x \chi) \partial_A \chi + g^{AB} \partial_A \chi \partial_B \chi \,.
\]
Rearranging terms, this gives
\begin{equation}\label{4V18.1}
\partial_x \chi + \frac{x}{2} (\partial_x \chi)^2 + \underbrace{\frac{xg^{xA}}{g^{xx}}}_{=O(x^{n+1})}(\frac{1}{x} + \partial_x \chi) \partial_A \chi + \underbrace{\frac{xg^{AB}}{2g^{xx}}}_{=O(x)} \partial_A \chi \partial_B \chi = \underbrace{\frac{x}{2g^{xx}}(1 - \frac{1}{x^2} g^{xx})}_{=O(x^{n-1})}
 \,.
\end{equation}
It is readily seen that the conformal infinity $\{x = 0\}$ is non-characteristic, and so existence of  a  function $\chi$ in a neighbourhood of the boundary follows.
Note that for a $C^{\ell}$-compactifiable metric the source term at the right-hand side of \eqref{4V18.1} is $C^{\ell-1}$ only, which results in a $C^{\ell-1}$ solution. But $\chi$ will be $C^{\ell|m-1}$ for metrics which are $C^{\ell|m}$.

It follows from \eqref{5V18.1}-\eqref{4V18.1} that
\begin{equation}\label{5V18.2}
  y = x + O(x^{n+1})
   \,.
\end{equation}

Now extend  local coordinate functions $x^A$ defined on the conformal infinity to  local coordinate functions $y^A$ defined in a neighborhood thereof, so that $y \mapsto (y,y^A)$ are geodesics with respect to the metric $y^2 g$, orthogonal to the conformal boundary. Since $y^2 g$ is $C^{\ell|m-1}$, $y^A$ is $C^{\ell|m-2}$. In the new coordinate system the metric $g$ is $C^{\ell-1|m-2}$--conformally compactifiable and takes on the desired form.

Let us denote by $\bar g_{ij} $ the metric coefficients of the metric $y^2 g$ in the Gauss coordinates above. From $\bar g_{yy}=1$ and  $\bar g_{yA}=0$ we obtain
\begin{equation}\label{5V18.3}
 y^2   g_{xA} = \frac{\partial y}{\partial x} \frac{\partial y}{\partial x^A}
  + \bar g_{BC} \frac{\partial y^B}{\partial x ^A} \frac{\partial y^C}{\partial x}
   \,.
\end{equation}
Letting $M^{A B}$ denote the matrix inverse to $ \bar g_{BC} \frac{\partial y^B}{\partial x ^A}$, this implies
\begin{equation}\label{5V18.4}
\frac{\partial y^A}{\partial x  }
 = M ^{AB}\left( y^2x^{-2} g_{xB} - \frac{\partial y}{\partial x} \frac{\partial y}{\partial x^B}
\right)
\,.
\end{equation}
Assuming that $m\ge 2$, integrating in $x$ and using \eqref{5V18.2} one obtains
\begin{equation}\label{5V18.6}
  y^A
	 = x^A + O(x^{n+1})
\,.
\end{equation}
If $\ell \geq n+1$ and $m \geq 4$, the invariance of the mass aspect follows now from Lemma~\ref{Lem:MAInv} using Taylor expansions.
\qed
\end{proof}

\subsection{Perturbations of infinity at order $x^{n-2}$,  generalised mass aspect function}
 \label{ss29X17.1}

In view of Lemma \ref{Lem:MAInv}, in order to  change  the mass aspect function via a coordinate transformation, one needs to work with metrics $g$ such that
\[
 \lambda = g - \mathring{g}
\]
does not satisfy \eqref{28VII17.4}. For our later purposes it suffices to consider the case that $\lambda$ decays ``one order slower'' than the decay given by \eqref{28VII17.4}.

As a by-product of our analysis, we will identify, in dimensions $n \geq 5$, a class of such metrics
where the mass equals the integral of a generalised mass aspect function \emph{which  can be changed by a coordinate transformation},
and which coincides with the mass aspect function when  \eqref{28VII17.4} holds. The point is that the mass integrand acquires new terms when the asymptotic coordinate conditions are relaxed, as compared to the ones in  {\eqref{18VIII15.6}-\eqref{18VIII15.8}}. This new integrand is the generalised mass aspect function. One can exploit the freedom gained, together with a subsequent deformation of the metric, to obtain a new nearby metric, with almost the same mass or energy-momentum, which satisfies again the more stringent conditions \eqref{28VII17.4} after the deformation but has now a different mass aspect function.

For $x > 0$ and $\ell = 0, 1, 2, \ldots$, define
\begin{equation}
\Omega_\ell(x) = \sup_{x^C \in \transversemanifold}   \sum_{0 \leq j \leq \ell}
 \,|\mathring{\nabla}^j \lambda(x,x^C)|_{\mathring{g}}
	\,.\label{Eq:30XI17.1}
\end{equation}

We proceed by inspecting the formula for the scalar curvature. To this end, define $\lambda^{ij}$ by raising the indices of $\lambda_{ij}$ with respect to $\mathring{g}$ and define its first order \emph{Newton tensor}
\[
T^{ij} = \lambda^{ij} - \tr_{\mathring{g}}(\lambda)\,\mathring{g}^{ij}.
\]
By \eqref{exp_Ricci},
\[
R[g] - R[\mathring{g}] = \znabla_i \znabla_j T^{ij} - \tr_{\mathring{g}}(T)
	+ O(\Omega_2^2(x))
		\,.
\]
Define $F^i = \znabla_j T^{ij}$. A direct computation gives
\begin{align*}
F^x
	&= \partial_x T^{xx} - \frac{n+1}{x} T^{xx} + \bzmcD _A T^{xA} + \frac{1}{x} \ringh_{AB} T^{AB} \\
		&\qquad\qquad - \frac{(n-1)kx}{2} T^{xx} + O(x^3)T^{xx} + O(x^3) \ringh_{AB}\,T^{AB}
			\,,\\
F^A
	&= \partial_x T^{xA} - \frac{n+2}{x} T^{xA} + \bzmcD _B T^{AB}\\
		&\qquad\qquad - \frac{(n+1)kx}{2} T^{xA} + O(x^3)T^{xB}
			\,,\\
\znabla_i F^i
	&= \partial_x F^x - \frac{n}{x} F^x + \bzmcD _A F^A - \frac{(n-1)kx}{2} F^x + O(x^3) F^x
	\,.
\end{align*}
Therefore
\begin{align}
R[g] - R[\mathring{g}]
	&= \znabla_i F^i - \tr_{\mathring{g}}(T) + O(\Omega_2^2(x))
		\nonumber\\
	&= x^{n+1}\partial_x \Big[x^{-1} \partial_x \left( x^{-n} T^{xx}  \right)
 \nonumber\\
			&\qquad\qquad\qquad\qquad + x^{-n-2} \ringh_{AB} T^{AB} + 2x^{-n-1} \bzmcD _A T^{xA}\Big]\nonumber\\
		&\qquad\qquad - (n-1)k x^{n+1}\partial_x(x^{-n} T^{xx}) - \frac{(n-2)k}{2} \ringh_{AB} T^{AB}\nonumber\\
		&\qquad\qquad - nkx \bzmcD _A T^{xA} + \bzmcD _A \bzmcD _B T^{AB}\nonumber\\
		&\qquad\qquad + O(\max(x^4\,\Omega_1(x), \Omega^2_2(x)))
			\,. \label{28VII17.8}
\end{align}

Now, suppose that we have a development of the $T^{xx}$, $T^{xA}$ and $T^{AB}$'s in series of powers of $x$ at $x = 0$, say starting at order $x^{n_1}$, $x^{n_2}$ and $x^{n_3}$ respectively (for some $n_1, n_2, n_3 \geq 0$). Observe that the contribution of the leading coefficients of $T^{xx}$, $T^{xA}$ and $T^{AB}$'s to the right-hand side of \eqref{28VII17.8} are of order $x^{n_1 - 2}$, $x^{n_2 -1}$ and $x^{n_3 - 2}$ respectively, except for the following four cases:

\begin{enumerate}[(i)]
\item If $n_1 = n$, there is no contribution from the leading coefficient of $T^{xx}$.
\item If $n_1 = n+2$,  the leading coefficient of $T^{xx}$ contributes a term of order $O(x^{n+2})$.
\item If $n_2 = n+1$,  the leading coefficients of $T^{xA}$ contribute a term of order $O(x^{n+2})$.
\item If $n_3 = n+2$,  the leading coefficients of $T^{AB}$ contribute a term of order $O(x^{n+2})$.
\end{enumerate}

On the other hand, in view of \eqref{Rabc3b}, $R[g] - R[\mathring{g}]$ should decay faster than $O(x^n)$. This leads us to consider metrics $g$ such that the tensor $T$ satisfies
\begin{align}
T^{xx}
	& = x^n\Tone{}^{xx}(x^C) + x^{n+2} \Ttwo{}^{xx}(x^C) + o_2(x^{n+2})
	\,,\label{28VII17.9xx}\\
T^{AB}
	& =  x^{n+2} \Ttwo{}^{AB}(x^C) + o_2(x^{n+2})
	\,,\label{28VII17.9AB}\\
T^{xA}
	& =  x^{n+1} \Tohalf{}^{xA}(x^C) + o_2(x^{n+1})
	\,.\label{28VII17.9xA}
\end{align}
Under these assumptions, the sum of the last three lines in \eqref{28VII17.8} is of order $O(x^{n+1})$,
 and the mass can be computed as the integral over ${\transversemanifold}$ of the leading term of the sum contained in the square bracket on the right-hand side of \eqref{28VII17.8}. For this class of metrics one can thus define the \emph{generalised mass aspect function} as
\begin{equation}
\Theta := -\frac{1}{n}\ringh_{AB} \Ttwo{}^{AB} - \frac{2}{n} \bzmcD _A\,\Tohalf{}^{xA} - \frac{2}{n} \Ttwo{}^{xx}
	\,. \label{28VII17.MAX}
\end{equation}
It should be clear that, when $\lambda$ satisfies \eqref{28VII17.4}, the above formula simplifies to \eqref{28VII17.6}.

We note that, under \eqref{28VII17.9xx}-\eqref{28VII17.9xA}, we have
\begin{equation}
\lambda_{xx} = O_{2}(x^{n-4})\,,\ \lambda_{AB} = O_{2}(x^{n-4}) \text{ and } \lambda_{xA} = O_{2}(x^{n-3})
	\,.
	\label{30VII17.1}
\end{equation}
If one asks that $\Omega_2(x)$ decays slightly better than $x^{\frac{n-2}{2}}$ (so that the mass can be defined), one is led  to the restriction $n \geq 5$.

In addition to the coordinate transformations already studied in Lemma \ref{Lem:MAInv}, there is another type of coordinate transformations which preserves the asymptotic behaviors \eqref{28VII17.9xx}-\eqref{28VII17.9xA}: $x \mapsto \bar x = x + x^{n-1}\,\psi(x^C)$. As we will now see, this change of variable leads to a change in the generalised mass aspect function.

\begin{lemma}\label{Lem:UChangeVar}
 Let $n \ge 3$. Assume that the metric $g$ satisfies \eqref{30VII17.1}.
Under the coordinate transformation
$$
(x,x^A) \mapsto (\bar x = x + x^{n-1}\,\psi(x^C),x^A)
	\,,
$$
with $\psi \in C^3(N^{n-1})$,
the tensor $T$ transforms as follows
\begin{align*}
T^{xx}
	&\rightarrow T^{xx}(\bar x,x^C) - 2(n-1)\bar x^{n}\,\psi(x^C) - k(n-1)\bar x^{n + 2}\,\psi(x^C)
			+ O(\bar x^{\min(n+4,2n-2)})
				\,,\\
T^{xA}
	&\rightarrow T^{xA}(\bar x,x^C) - \bar x^{n+1}\,\bzmcD ^A \psi(x^C)
            +  O(\bar x^{\min{}(n+3,2n-1)})
		\,,\\
T^{AB}
	&\rightarrow T^{AB}(\bar x,x^C) - k(n-2) \bar x^{n+2}\,\psi(x^C)\,\ringh^{AB}(x^C)+ O(\bar x^{\min(n+4,2n-2)})
		\,.
\end{align*}
\end{lemma}

This lemma will not be needed in our main results,  we therefore defer its proof to Appendix~\ref{A28XII17.1}.

\begin{corollary}\label{Cor:UChangeVarMA}
 Let $n \ge 5$. Then the coordinate transformation
$$
(x,x^A) \mapsto (\bar x = x + x^{n-1}\,\psi(x^C),x^A)
$$
preserves the asymptotic conditions \eqref{28VII17.9xx}-\eqref{28VII17.9xA}. Furthermore, the new and old generalised mass aspect functions defined in \eqref{28VII17.MAX} are related by
\[
\Theta_{\mathrm{new}} = \Theta_{\mathrm{old}} + \frac{2}{n} \bzmcD ^A\bzmcD _A \psi + k(n-1)\psi.
\]
\end{corollary}

\subsection{Perturbation of infinity at order $\frac{x^{n-2}}{|\ln x|}$}
\label{ss22XII17.1}

The requirement that $n$ be at least five in the previous subsection is quite restrictive. Furthermore, the most direct application of the results in that subsection to the proof of our deformation theorem will introduce a perturbation of order $O(x^{n-1})$ in the $T^{xA}$ component, which contributes an error estimation of order $O(\max(x^{n+1}, x^{2n-6}))$ in the scalar curvature (cf. \eqref{28VII17.8}), which does not decay fast enough to ensure the integrability condition \eqref{Rabc3b} in dimensions $n = 4, 5, 6$.
{By arranging a suitable form for $T^{xA}$ and making an appropriate change of the angular variables, the error estimation can improved to $O(\max(x^{n+2}, x^{2n-4}))$, which takes care of dimensions $n = 5,6$. In dimension $n = 4$, we circumvent the above complication by working with metrics  which are, roughly speaking, perturbations of infinity at order $O(\frac{x^{n-2}}{|\ln x|})$, which is slightly milder than that in the previous subsection.

For this,
let $x \mapsto \Xi(x)$ be a smooth function which is defined for small positive values of $x$ and satisfies for some $\ell \geq 0$ that
\begin{equation}\label{18XII17.31}
 \Lambda_{\ell + 1}(x)  = O(1)
 \,,
\ \textrm{where}
\
\Lambda_{\ell +1}(x) := \sum_{0 \leq l \leq \ell +1} x^l\,|\partial_x^l \Xi|
	\,.
\end{equation}
In the notation of \eqref{18XII17.22}, since $\Xi$ depends only upon $x$ it holds that $\Xi(x) = O_j(\Lambda_j(x) )$ for any $0\le j \le \ell+1$. (For readers who would like to zoom ahead to the proof of our deformation theorem, $\Xi$ will be chosen so that $\Lambda_\ell(x) = O(\frac{x^{n-2}}{|\ln x|})$, but we do not assume this in the present section.)

We have:

\begin{lemma}
 \label{Lem:UChangeVarAng}
Let $n \ge 3$,  $\ell \geq 0$, and suppose that \eqref{18XII17.31} holds.
Assume that there is a positive function $x \mapsto \Upsilon(x)$
such that, for small values of $x$, it holds
\begin{eqnarray}
&
\nn
\lambda_{xx} = O_\ell(x^{-4} \Upsilon(x))\,,
 \ \lambda_{AB} = O_\ell(x^{-4} \Upsilon(x))\,, \
  \lambda_{xA} = O_\ell(x^{-5} \Upsilon(x))
   \,,
&
\\
&\zmcD _A\lambda_{xx} =  O_\ell(x^{-4} \Upsilon(x))\,,
 \ \zmcD _C\lambda_{AB} =  O_\ell(x^{-4} \Upsilon(x))\,,
   \text{ and } \zmcD _B\lambda_{xA} =  O_\ell(x^{-5} \Upsilon(x))
	\,.&
 \nn
\\
 &&
 \label{26XI17.1}
\end{eqnarray}
Let $X\equiv X^A\partial_A$ be a
$C^{\ell+1}$
section of $T\transversemanifold$.
Then, under the coordinate transformation,
\begin{equation}\label{28XII17.1}
(x,x^A) \mapsto \big(x,y^A = x^A + \Xi(x)\,X^A(x^C)\big)
	\,,
\end{equation}
the tensor $T$ transforms as follows
\begin{align*}
T^{xx}
	&\rightarrow T^{xx}(x,y^C) + 2 x^2\,\Xi(x)\,\bzmcD _A X^A(y^C)
			+ O_{\ell}(x^2 \Lambda_{\ell+1}^2(x) +  \Lambda_{\ell+1}^3(x)+ x^{-2}\, \Upsilon(x)\,\Lambda_{\ell}(x))
				\,,\\
T^{xA}
	&\rightarrow T^{xA}(x,y^C) - x^2 \Xi'(x)\,X^A (y^C)
		  + O_{\ell}(x \Lambda_{\ell+1}^2(x) + x^{-{1}}\Upsilon(x)\,\Lambda_{\ell}(x))
   \,,\\
T^{AB}
	&\rightarrow T^{AB}(x,y^C) - x^2\,\Big(1 - \frac{k}{4}x^2\Big)^{-2}\,\Xi(x)\,(\bzmcD ^A X^B + \bzmcD ^B X^A - 2\bzmcD _D X^D\,\ringh^{AB})(y^C)\\
		&\qquad\qquad - x^2  [\Xi'(x)]^2\,|X|_{\mathring h}^2\,\ringh^{AB}
			+ O_{\ell}(x^2 \Lambda_{\ell+1}^2(x) +  \Lambda_{\ell+1}^3(x)+ x^{-2}\,\Upsilon(x)\,\Lambda_{\ell}(x))
		\,,
\end{align*}
where the implicit constants in the big $O$ terms depend only on the implicit constants in
 \eqref{18XII17.31} and \eqref{26XI17.1},  $\|X\|_{C^{{\ell + 1}}(\transversemanifold)}$, $n$ and $\ell$.
\end{lemma}

\begin{proof} In the new coordinate system $(x, y^A)$, we will use $\tilde \lambda$ to denote the difference between $g$ and the new reference metric
$$\mathring{\tilde g} = x^{-2}\Big[\mathrm{d}x^2 + \Big(1 - \frac{k}{4} x^2\Big)^2\ringh_{AB}(y^C)\,\mathrm{d}y^A\,\mathrm{d}y^B
 \Big]
 \,.
$$
We will accordingly use a tilde to refer to the metric components of $\tilde \lambda$, its Newton tensor etc.

Define the matrix $ M \equiv (M^A{}_B) $ by
\begin{equation}\label{18XII17.1}
   M^A{}_B   \equiv
 M^A{}_B(x, y^C)
 =
 \delta^A_B + \Xi(x) \frac{\partial X^A}{\partial x^B}\big(x^D(y^C)\big)
 \,.
\end{equation}
We have
 \begin{align*}
\mathrm{d}x^A
 &=
  (M^{-1})^A{}_B\left(
 \mathrm{d}y^B - \Xi'(x) X^B\big(x^D (y^C)\big)\mathrm{d}x
 \right)
\\
	&= [-\Xi'(x)\,X^A(y^C) +
 O_{\ell}(x^{-1}\Lambda_{\ell+1}^2(x))
	] \mathrm{d}x \\
		&\qquad\qquad  + [\delta^A{}_B - \Xi(x)\,\partial_B X^A(y^C)
			+ O_{\ell}(\Lambda_{\ell}^2(x))
		]\mathrm{d}y^B
		\,,
\\
\ringh_{AB}(x^C)
	&= \ringh_{AB}(y^C) - \Xi(x) \partial_{D} \ringh_{AB}(y^C)\,X^D(x^C) + O_{\ell}( \Lambda_\ell^2(x) )
\\
	&= \ringh_{AB}(y^C) - \Xi(x) \partial_{D} \ringh_{AB}(y^C)\,X^D( y^C ) +O_{\ell}( \Lambda_{{\ell}}^2(x) )
	\,.
\end{align*}
This implies that
\begin{align*}
&\ringh_{AB}(x^C)\,\mathrm{d}x^A\,\mathrm{d}x^B\\
	&\qquad = \ringh_{AB}(y^C)\,\mathrm{d}y^A\,\mathrm{d}y^B  +[\Xi'(x)]^2\,|X|_{{\mathring h}}^2(y^C) \mathrm{d}x^2 \\
		&\qquad\qquad + [\Xi'(x)]^2\,h_{AB}(y^C)\,X^A(y^C)\,X^B(y^C)\,\mathrm{d}x^2\\
		&\qquad\qquad - 2\Xi'(x) \ringh_{AB}(y^C)\,X^B(y^C)\,\mathrm{d}x\,\mathrm{d}y^A\\
		&\qquad\qquad - \Xi(x)\,\underbrace{[\partial_D \ringh_{AB}\,X^D + 2\ringh_{D(A}\partial_{B)} X^D]}_{=\bzmcD _B X_A + \bzmcD _A X_B}(y^C)\,\mathrm{d}y^A\,\mathrm{d}y^B
\\
&\qquad\qquad
			+ O_{\ell}(x^{-2}\,\Lambda_{\ell+1}^3(x))\mathrm{d}x^2  + O_{\ell}(x^{-1}\,\Lambda_{\ell+1}^2(x)) \mathrm{d}x\,\mathrm{d}y^A
 + O_{\ell}(\Lambda_{{ \ell }}^2(x) )
 \mathrm{d}y^A\,\mathrm{d}y^B
 \,,
\end{align*}
where $X_A = \ringh_{AB}X^B$.
Using the trivial identity
\begin{equation}\label{20XII17.1}
 \lambda(x,x^D (y^C))= \lambda(x,y^C )- \int_0^1\frac{d\Big(\lambda\big(x,x^E (y^C) + s \Xi(x)X^E(x^D (y^C))\big)\Big)}{ds} ds
\end{equation}
to replace every occurrence of  $ \lambda(x,x^D(y^C))$ by $\lambda(x,y^C ) $, together with the hypothesis \eqref{26XI17.1} to estimate the associated error terms, a calculation gives
\begin{align}
 \nn
\tilde \lambda
	&= \Big[\lambda_{xx}(x,y^C)  + x^{-2}\,\Big(1 - \frac{k}{4}x^2\Big)^2\, [\Xi'(x)]^2\,|X|_{{\mathring h}}^2(y^C) \Big]\, \mathrm{d}x^2
\\
 \nn
		&\qquad + 2\Big[\lambda_{xA}(x,y^C) - x^{-2}\,\Big(1 - \frac{k}{4}x^2\Big)^2\Xi'(x)\,X_A(y^C) \Big] \,\mathrm{d}x\mathrm{d}y^A\nonumber
\\
 \nn
		&\qquad + \Big[\lambda_{AB}(x,y^C)  -  x^{-2}\,\Big(1 - \frac{k}{4}x^2\Big)^2\, \Xi(x)(\bzmcD _B X_A + \bzmcD _A X_B)(y^C)\Big]\,\mathrm{d}y^A\mathrm{d}y^B\nonumber
\\
 \nn
		&\qquad + O_{\ell}(x^{-4}\,\Lambda_{\ell+1}^3(x)+ x^{-{6}}\,\Upsilon(x)\,\Lambda_{\ell}(x))\mathrm{d}x^2
\\
 \nn
		&\qquad + O_{\ell}(x^{-3} \Lambda_{\ell+1}^2(x)  + x^{-{5}}\,\Upsilon(x)\,\Lambda_{\ell}(x)) \mathrm{d}x\,\mathrm{d}y^A
\\
		&\qquad + O_{\ell}( x^{-2} \Lambda_{\ell+1}^2(x)  + x^{-4}\,\Upsilon(x)\, \Lambda_{\ell}(x))\mathrm{d}y^A\,\mathrm{d}y^B
			\,.
 \label{20XII17.1a}
\end{align}

We next compute the tensor $\tilde T$:
\begin{align*}
\tr_{\mathring{\tilde g}}(\tilde \lambda)
	& = x^2 \,\tilde \lambda_{xx} +  x^2 \Big(1 - \frac{k}{4} x^2\Big)^{-2} \,\ringh^{AB}\,\tilde \lambda_{AB}\\
	& =  \tr_{\mathring{g}}(\lambda)
		 + \Big(1 - \frac{k}{4}x^2\Big)^2\, [\Xi'(x)]^2\,|X|_{{\mathring h}}^2 - 2\Xi(x)\,\bzmcD _A X^A
				\\
				&\qquad\qquad  + O_{\ell}(\Lambda_{\ell+1}^2(x) + x^{-2}\,\Lambda_{\ell+1}^3(x)+ x^{-4}\,\Upsilon(x)\,\Lambda_{\ell}(x))
				\,,\\
\tilde T^{xx}
	&=  x^4 \,\tilde \lambda_{xx} - \tr_{\mathring{\tilde g}}(\tilde \lambda)\,x^2\\
	&= T^{xx} + 2x^2\,\Xi(x)\,\bzmcD _A X^A
			+ O_{\ell}(x^2\Lambda_{\ell+1}^2(x) + \Lambda_{\ell+1}^3(x)+ x^{-2}\,\Upsilon(x)\,\Lambda_{\ell}(x))
				\,,
\\
\tilde T^{xA}
	&= x^4\Big(1 - \frac{k}{4}x^2\Big)^{-2} \ringh^{AB} \tilde \lambda_{xB}\\
	&= T^{xA} - x^2 \Xi'(x)\,X^A
		+ O_{\ell}(x\,\Lambda_{\ell+1}^2(x)+ x^{-{1}}\,\Upsilon(x)\,\Lambda_{\ell}(x))
		\,,\\
\tilde \lambda^{AB}
	&= x^4\Big(1 - \frac{k}{4}x^2\Big)^{-4}\,\ringh^{AC}\,\ringh^{BD} \tilde \lambda_{CD}\\
	&= \lambda^{AB} - x^2\,\Big(1 - \frac{k}{4}x^2\Big)^{-2}\Xi(x)\,(\bzmcD ^A X^B + \bzmcD ^B X^A)
		\\
		&\qquad\qquad +O_{\ell}(x^2\Lambda_{\ell+1}^2(x)  + \Upsilon(x)\,\Lambda_{\ell}(x))
		\,,\\
\tilde T^{AB}
	&= \tilde \lambda^{AB} - x^2\Big(1 - \frac{k}{4}x^2\Big)^{-2}\tr_{\mathring{\tilde g}}(\tilde \lambda)\,\ringh^{AB}\\
	&= T^{AB}  - x^2  [\Xi'(x)]^2\,|X|_{{\mathring h}}^2\,\ringh^{AB}\\
		&\qquad\qquad - x^2\,\Big(1 - \frac{k}{4}x^2\Big)^{-2}\Xi(x)\,(\bzmcD ^A X^B + \bzmcD ^B X^A - 2 \bzmcD _C X^C\,\ringh^{AB})\\
		&\qquad\qquad  + O_{\ell}(x^2 \Lambda_{\ell+1}^2(x) +  \Lambda_{\ell+1}^3(x)+ x^{-2}\,\Upsilon(x)\,\Lambda_{\ell}(x))
		\,.
\end{align*}
This completes the proof.
\end{proof}
\qedskip

We now derive a version of \eqref{28VII17.8} where the mixed terms $T^{xA}$ are allowed to decay slower than the $T^{xx}$ and $T^{AB}$ terms. More precisely, we assume, for some smooth function $\Xi$ and vector field $X = X^A\partial_A$ (on $\transversemanifold$), that $T$ can be expressed as a sum of $\TohalfX(x,x^C) = x^2\,\Xi'(x)\,X^{A}(x^C)\,(\partial_x \otimes \partial_{A} + \partial_{A} \otimes \partial_x)$ and terms which decay faster than $\TohalfX$. For $\ell \geq 0$, let
\begin{align}
\OmegatohalfX_\ell(x)
	&= \sup_{x^C \in \transversemanifold}   \sum_{0 \leq j \leq \ell}
 \,[|\znabla^j (T - \TohalfX)|_{\mathring{g}} + x^{-1}\,|\znabla^j \zmcD(T - \TohalfX)|_{\mathring{g}}](x,x^C)
	\,,\label{Eq:18XII17.1}\\
\Upsilon_\ell(x)
	&= x^2(\OmegatohalfX_\ell(x) + \Lambda_{\ell+1}(x))
	\,,\label{Eq:18XII17.2}
\end{align}
where $\Lambda_{\ell+1}$ is as defined in \eqref{18XII17.31}.
Note that \eqref{26XI17.1} then holds with $\Upsilon = \Upsilon_\ell$.

\begin{corollary}\label{Cor:3.10Extended}
Let $n \geq 3$.
Assume that there exist a smooth vector field $X = X^{A}\partial_A $ on $ \transversemanifold$ and a smooth function $x \mapsto \Xi(x)$ such that
\begin{equation}
\Upsilon_2(x) = O(x^4)
	\,,\label{Eq:22XII17.E1}
\end{equation}
where $\Upsilon_\ell$ is as defined in \eqref{Eq:18XII17.2}.
After the change of coordinates
\begin{equation}\label{26XII17.1}
 (x,x^A) \mapsto (x,y^A = x^A + \Xi(x)\,X^A(x^C))
\end{equation}
one has
\begin{align}
R[g] - R[\mathring{g}]
	&= x^{n+1}\partial_x \Big[x^{-1} \partial_x \left( x^{-n} T^{xx}  \right)
 \nonumber\\
			&\qquad\qquad\qquad\qquad + x^{-n-2} \ringh_{AB} T^{AB} + 2x^{-n-1} \bzmcD _A T^{xA}\Big]\nonumber\\
		&\qquad \qquad  -(n-1)x^{n+1}\partial_x \Big[x^{-n}[\Xi'(x)]^2 \,|X|_{{\mathring h}}^2\Big]  \nonumber\\
		&\qquad\qquad - (n-1)k x^{n+1}\partial_x(x^{-n} T^{xx}) - \frac{(n-2)k}{2} \ringh_{AB} T^{AB}\nonumber\\
		&\qquad\qquad - nkx \bzmcD _A T^{xA} + \bzmcD _A \bzmcD _B T^{AB}\nonumber \\
		&\qquad \qquad  + O(\max(x^2\Upsilon_1(x), x^{-4}\Upsilon_2^2(x))
			\,. \label{28VII17.8Extended}
\end{align}
where the implicit constant in the error term depends on $n$, $\|X\|_{C^3(\transversemanifold)}$ and the implicit constant in \eqref{Eq:22XII17.E1}.
\end{corollary}

\begin{remark}\label{Rem:3.10App}
Note that while $|T - \TohalfX|_{\mathring{g}} = O_\ell(x^{-2}\Upsilon_\ell(x))$, one has $|\TohalfX|_{\mathring{g}} = O_\ell(x^{-3}\Upsilon_\ell(x))$ and so $|T|_{\mathring{g}} = O_\ell(x^{-3}\Upsilon_\ell(x))$.
 Thus, if one attempts to apply directly formula \eqref{28VII17.8}, one obtains an error estimation of order $O(\max(x\Upsilon_1(x), x^{-6}\,\Upsilon_2^2(x)))$, which is larger than that in \eqref{28VII17.8Extended}.
\end{remark}

\begin{proof}
By Lemma \ref{Lem:UChangeVarAng},
the change of angular variables \eqref{26XII17.1}
leads to better decay properties. Namely, with respect to the new coordinate system $(x,y^A)$, we have that $\tilde T^{xx}$, $\tilde T^{xA}$ and $\tilde T^{AB}$ are of order $O_\ell(\Upsilon_{\ell}(x))$, for $\ell = 1,2$.
Using $
\Lambda_{\ell +1} = O(x^{-2}\Upsilon_\ell)
$
and $\Upsilon_2 = O(x^4)$
(``$x$ to power four'', not to be confused with the coordinate ``$x$ subscript four'')
we obtain
$$
 \Lambda_{\ell +1}^3 = O(x^{-6}\Upsilon_\ell^3) = O(x^{-4} \Upsilon_\ell x^{-2} \Upsilon_\ell^2)= O(x^{-2}\Upsilon_\ell^2)
  \,,
$$
which will eventually be estimated as $O(x^{-4}\Upsilon_\ell^2)$.
 This can be used to rewrite the error terms in the conclusions of Lemma~\ref{Lem:UChangeVarAng} as follows:
\begin{align}
\tilde T^{xx}(x,y^C)
	&= T^{xx}(x,y^C) + 2 x^2\,\Xi(x)\,\bzmcD _A X^A(y^C)
			+ O_\ell(x^{-4}\Upsilon_{\ell}^2(x))
				\,,\label{Eq:16XII17-X1}\\
\tilde T^{xA}(x,y^C)
	&= T^{xA}(x,y^C) - \underbrace{x^2 \Xi'(x)\,X^A (y^C)}_{=\TohalfX{}^{xA}(x,y^C)}
		  + O_\ell(x^{-4}\Upsilon_{\ell}^2(x))
   \,,\label{Eq:16XII17-X2}\\
\tilde T^{AB}(x,y^C)
	&= T^{AB}(x,y^C) - x^2\,\Big(1 - \frac{k}{4}x^2\Big)^{-2}\,\Xi(x)\,(\bzmcD ^A X^B + \bzmcD ^B X^A - 2\bzmcD _D X^D\,\ringh^{AB})(y^C)\nonumber\\
		&\qquad\qquad - x^2  [\Xi'(x)]^2\,|X|_{{\mathring h}}^2\,\ringh^{AB}
			+ O_\ell(x^{-4}\Upsilon_{\ell}^2(x))
		\,,\label{Eq:16XII17-X3}
\end{align}
where we are using the same notations as in the proof of Lemma~\ref{Lem:UChangeVarAng} and we have used that $x^{-4}\,\Upsilon_2(x) = O(1)$.

Now, by \eqref{28VII17.8} \emph{in the new coordinates},
 we have
\begin{align*}
R[g] - R[\mathring{g}]
	&= x^{n+1}\partial_x \Big[x^{-1} \partial_x \left( x^{-n} \tilde T^{xx}  \right)
 \nonumber\\
			&\qquad\qquad\qquad\qquad + x^{-n-2} \ringh_{AB} \tilde T^{AB} + 2x^{-n-1} \bzmcD _{\partial_{y^A}} \tilde T^{xA}\Big]\nonumber\\
		&\qquad\qquad - (n-1)k x^{n+1}\partial_x(x^{-n} \tilde T^{xx}) - \frac{(n-2)k}{2} \ringh_{AB} \tilde T^{AB}\nonumber\\
		&\qquad\qquad - nkx \bzmcD _{\partial_{y^A}} \tilde T^{xA} + \bzmcD _{\partial_{y^A}} \bzmcD _{\partial_{y^B}} \tilde T^{AB}
		+ O(\max(x^{2} \Upsilon_1(x) , x^{-4} \Upsilon_2^2(x))
		\,.
\end{align*}
Therefore, by \eqref{Eq:16XII17-X1}-\eqref{Eq:16XII17-X3},
\begin{align*}
R[g] - R[\mathring{g}]
	&= x^{n+1}\partial_x \Big\{x^{-1} \partial_x \left( x^{-n} (T^{xx} + 2x^2 \Xi(x) \bzmcD _D X^D ) \right)
 \nonumber
 \\
			&\qquad\qquad\qquad + x^{-n-2}  \Big[\ringh_{AB}T^{AB} + 2(n-2)\,x^2\,\Big(1 + \frac{k}{2}x^2\Big)\,\Xi(x)\,\bzmcD _D X^D
\\
					&\qquad\qquad\qquad\qquad - (n-1)x^2  [\Xi'(x)]^2\,|X|_{{\mathring h}}^2 \Big]
\\
			&\qquad\qquad\qquad  + 2x^{-n-1} \bzmcD _A (T^{xA} - x^2\,\Xi'(x)\,X^A) \Big\}\nonumber
\\
		&\qquad\qquad - (n-1)k x^{n+1}\partial_x(x^{-n} (T^{xx} + 2x^2\, \Xi(x) \bzmcD _D X^D)
\\
		&\qquad\qquad - \frac{(n-2)k}{2}
 \big(
  \ringh_{AB}T^{AB} + 2(n-2)x^2\,\Xi(x)\,\bzmcD _D X^D
   \big)
   \nonumber
\\
		&\qquad\qquad - nkx \bzmcD _A (T^{xA}  - x^2\,\Xi'(x)\, X^A)
\\
		&\qquad\qquad+ \bzmcD _A \bzmcD _B
 \big(
  T^{AB} - x^2\,\Xi(x)\,(\bzmcD ^A X^B + \bzmcD ^B X^A - 2\bzmcD _D X^D\,\ringh^{AB})
\, \ringh^{AB}
    \big)
\\
		&\qquad\qquad + O(\max(x^{2} \Upsilon_1(x) , x^{-4}  \Upsilon_2^2(x))
\\
	&= x^{n+1}\partial_x \Big[x^{-1} \partial_x \left( x^{-n} T^{xx} \right)
 \nonumber
\\
			&\qquad\qquad\qquad  + x^{-n-2} \ringh_{AB} T^{AB} + 2x^{-n-1} \bzmcD _A T^{xA}
\\
			&\qquad\qquad\qquad - (n-1) x^{-n}  [\Xi'(x)]^2\,|X|_{{\mathring h}}^2\, \Big]\nonumber
\\
		&\qquad\qquad - (n-1)k x^{n+1}\partial_x(x^{-n} T^{xx}) - \frac{(n-2)k}{2} \ringh_{AB} T^{AB}\nonumber
\\
		&\qquad\qquad - nkx \bzmcD _A T^{xA} + \bzmcD _A \bzmcD _B T^{AB}
\\
		&\qquad\qquad + O(\max(x^{2} \Upsilon_1(x) , x^{-4} \Upsilon_2^2(x))
		\,,
\end{align*}
where we have used that $\bzmcD _A\bzmcD _B \bzmcD ^A X^B - \bzmcD _A \bzmcD ^A \bzmcD _B X^B = (n-2)k\bzmcD _D X^D$ thanks to $R[\ringh]_{AB}=(n-2)k\ringh_{AB}$.
\end{proof}
\qedskip

\subsection{Hyperbolic symmetries}\label{ssec:HypSym}

In this  section  we assume that the transverse manifold $\transversemanifold$ is the standard sphere $\mathbb{S}^{n-1}$. In this case, $\mathring{g}$ is the hyperbolic metric on the hyperbolic space $\mathbb{H}^n$, and so has a large group of symmetries.

Consider the realization of the hyperbolic space $\mathbb{H}^n$ by the hyperboloid $\{(y_0, y) \in \mathbb{R}^{1 + n}: y_0^2 - |y|^2 = 1\}$ in the Minkowski space $\mathbb{R}^{1+n}$, where $|\cdot|$ denotes the standard Euclidean norm. Writing $y_0 = \cosh s$ and $y = \sinh s\,\theta$ for some $s \in \mathbb{R}_{\geq 0}$ and $\theta \in \mathbb{S}^{n-1}$, we obtain
\[
\mathring{g} = ds^2 + \sinh ^2 (s)\,\ringh
	\,,
\]
where $\ringh$ is the round metric on $\mathbb{S}^{n-1}$.
This can be brought to the form $\mathring{g} = x^{-2}(\mathrm{d}x^2 + (1 - \frac{1}{4}x^2)^2\ringh)$ considered earlier via the transformation $s = - \ln \frac{x}{2}$.

The group $SO(1,n)$ acts isometrically on $\mathbb{H}^n$. Consider a hyperbolic element of $SO(1,n)$ of the form
\[
M_{\alpha, e}:= \left[\begin{array}{cc}
\cosh\alpha & \sinh \alpha\,e^T\\
\sinh \alpha \,e & I_n + (\cosh \alpha - 1) e\otimes e^T
\end{array}\right]
	\,,
\]
where $ \alpha \in \mathbb{R}, e \in \mathbb{S}^{n-1}$ and $I_n$ denotes the $n \times n$ identity matrix. The transformation $(y_0 = \cosh s, y) = M_{\alpha,e} \cdot (\tilde y_0 = \cosh \tilde s, \tilde y)$ is given by
\begin{align*}
\cosh s
	&= \cosh \alpha \cosh \tilde s + \sinh \alpha \,e \cdot \tilde y
		\,,\\
 y
	&= (\sinh \alpha \,\cosh \tilde s + \cosh \alpha\, e \cdot \tilde y)e + ( \tilde y - e \cdot  \tilde y\, e)
		\,.
\end{align*}
Observe that $\theta:= \frac{1}{|y|} y$ and $\tilde \theta:= \frac{1}{|\tilde y|}\tilde y$ are related by
\[
\theta = \Phi_\infty(\tilde\theta) + O(\tilde x^2)
\]
where $\Phi_\infty$ is a transformation of $\mathbb{S}^{n-1}$ (which is viewed as the `boundary' of $\mathbb{H}^n$) given by
\[
\Phi_{\alpha,e}^\infty(\tilde\theta) = \frac{(\sinh \alpha  + \cosh \alpha \, e \cdot \tilde \theta)e + ( \tilde\theta - e \cdot  \tilde\theta\, e)}{\cosh \alpha + \sinh \alpha \,e \cdot \tilde\theta}
	\,.
\]
The inverse of $\Phi_{\alpha,e}^\infty$ is $\Phi_{-\alpha,e}^\infty$, i.e.
\[
(\Phi_{\alpha,e}^{\infty})^{-1}(\theta) = \frac{(-\sinh \alpha  + \cosh \alpha \, e \cdot  \theta)e + (\theta - e \cdot  \theta e)}{\cosh \alpha - \sinh \alpha \,e \cdot \theta}
	\,.
\]

Note that $\Phi_\infty$ is a conformal transformation of $\mathbb{S}^{n-1}$. To see this, let $\tilde x^A$ be a local coordinate system on some open subset $U \subset \mathbb{S}^{n-1}$, and let $x^A$ be a local coordinate system in $\Phi_\infty(U)$, and write $\ringh = \ringh_{AB}(x^C) \,\mathrm{d}x^A\,\mathrm{d}x^B = \mathring{\tilde h}_{AB}(\tilde x^C)\,d\tilde x^A\,d\tilde x^B$. As $M_{\alpha, e}$ is an isometry of $\mathbb{H}^n$, we have
\[
\mathring{g}
	= ds^2 + \sinh^2s\,\mathring{h}_{AB}(x^C + O(\tilde x^{2}))\,\mathrm{d}x^A\,\mathrm{d}x^B
	= d\tilde s^2 + \sinh^2 \tilde s\,\mathring{\tilde h}_{AB}(\tilde x^C)\,d\tilde x^A\,d\tilde x^B.
\]
It follows that
\begin{align}
\mathring{\tilde h}_{AB}(\tilde x^C)\,d\tilde x^A\,d\tilde x^B
	 &= \lim_{\tilde s \rightarrow \infty} \frac{\sinh^2s}{\sinh^2\tilde s}\,\mathring{h}_{AB}(x^C)\,\mathrm{d}x^A\,\mathrm{d}x^B
	 	\nonumber\\
	 & = (\cosh \alpha + \sinh \alpha \,e \cdot \tilde\theta)^2\,\mathring{h}_{AB}(x^C)\,\mathrm{d}x^A\,\mathrm{d}x^B
	 \,.\label{9VIII17.C1}
\end{align}

Consider now a metric of the form {\eqref{18VIII15.6}-\eqref{18VIII15.8}}, i.e.
\begin{align*}
g
	&= \mathring{g} + x^{n-2} \mu_{AB}(x^C)\,\mathrm{d}x^A\,\mathrm{d}x^B +o(x^{n-2}) \mathrm{d}x^i \mathrm{d}x^j\\
   	&=  \mathring{g} + \tilde x^{n-2} \tilde\mu_{AB}(\tilde x^C)\,d\tilde x^A\,d\tilde x^B +o(\tilde x^{n-2}) \mathrm{d}\tilde x^i \mathrm{d}\tilde x^j
	\,.
\end{align*}

Set $\mu = \mu_{AB}(x^C)\,\mathrm{d}x^A\,\mathrm{d}x^B$ and $\tilde \mu = \tilde\mu_{AB}(\tilde x^C)\,d\tilde x^A\,d\tilde x^B$. We have
\[
\tilde \mu = \Big(\lim_{x \rightarrow 0}\frac{x^{n-2}}{\tilde x^{n-2}}\Big)(\Phi_{\alpha, e}^\infty)^* \mu = \frac{1}{(\cosh \alpha + \sinh\alpha\,e \cdot\tilde\theta)^{n-2}}\,(\Phi_{\alpha, e}^\infty)^* \mu.
\]
Also, by \eqref{9VIII17.C1}, $\mathring{\tilde h} = (\cosh \alpha + \sinh\alpha\,e \cdot\tilde\theta)^2\,(\Phi_{\alpha, e}^\infty)^* \ringh$. It follows that
\[
\mathring{\tilde h}^{AB}\,\tilde\mu_{AB} = \frac{1}{(\cosh \alpha + \sinh\alpha\,e \cdot\tilde\theta)^{n}} (\ringh^{AB}\,\mu_{AB}) \circ \Phi_{\alpha, e}
	\,.
\]
In order words, the mass aspect functions $\Theta$ and $\tilde\Theta$ relative to the $(x, x^A)$ and $(\tilde x, \tilde x^A)$ coordinate systems are related by
\begin{equation}
\tilde\Theta(\tilde\theta) = \frac{1}{(\cosh \alpha + \sinh\alpha\,e \cdot\tilde\theta)^{n}} \Theta \circ \Phi_{\alpha, e}(\tilde\theta).
	\label{9VIII17.C2}
\end{equation}

Recall that the covector $(m_0 , m_1 \ldots, m_n)$ is defined by  \eqref{9VIII17.Add1},
\begin{align*}
m_0
	&= \int_{\mathbb{S}^{n-1}} \Theta(\theta)\,dv_{\ringh}(\theta)
		\,,\\
m_i
	&= \int_{\mathbb{S}^{n-1}}\Theta(\theta)\,\theta_i\,dv_{\ringh}(\theta)
		\,,
\end{align*}
where $\theta_i \equiv \theta^i$. Using \eqref{9VIII17.C1}, \eqref{9VIII17.C2} and the relation
$$
\frac{1}{\cosh \alpha + \sinh\alpha\,e \cdot\tilde\theta} = \cosh \alpha - \sinh\alpha\,e \cdot\theta
	\,,
$$ we find that the corresponding covector $(\tilde m_0, \tilde m_1, \ldots, \tilde m_n)$ relative to the tilde coordinate system is given by
\begin{align*}
\tilde m_0
	&= \int_{\mathbb{S}^{n-1}} \Theta(\theta)\,(\cosh \alpha - \sinh\alpha\,e \cdot\theta)dv_{\ringh}(\theta)\\
	&= \cosh\alpha \,m_0 - \sinh\alpha\,e \cdot \stackrel{\rightarrow}{\ellm }
		\,,\\
\tilde m_i
	&= \int_{\mathbb{S}^{n-1}}\Theta(\theta)\,[(-\sinh \alpha  + \cosh \alpha \, e \cdot  \theta)e_i + (\theta_i - e \cdot  \theta e_i)]\,dv_{\ringh}(\theta)\\
	&= (-\sinh \alpha\,m_0  + \cosh \alpha \, e \cdot \stackrel{\rightarrow}{\ellm })e_i  + ( m_i - e \cdot \stackrel{\rightarrow}{\ellm } \ e_i)
		\,,
\end{align*}
where $e_i = e^i$ and $e \cdot \stackrel{\rightarrow}{\ellm } \,= \sum_{i = 1}^n e^i  m_i$. In particular, this gives the well-known relation
\[
(m_0, m_1 \ldots, m_n) = M_{\alpha,e} \cdot (\tilde m_0, \tilde m_1, \ldots, \tilde m_i).
\]

\section{Proof of the deformation theorem}
 \label{s17VIII15.3}

We are ready now to formulate, and prove, a precise version of Theorem~\ref{T31VII17.1a}. We consider a metric $g$ which, on $\{0<x<x_0\}$, for some $x_0<1$, takes the form \eq{18VIII15.6}-\eq{18VIII15.8} with all tensors twice-differentiable.
We further suppose
that
\begin{equation}\label{21XI17.E1}
x^{-1} (R[g] - R[\mathring{g}])\in L^1(M)
	\,,
\end{equation}
and that there exist constants $C_1$ and $\alpha>0$ and  such that
\begin{equation}\label{18VIII15.9}
 \sum_{0 \leq l \leq 2}[ |\znabla^l \lambda|_{{\zg}}+ x^{-1} |\znabla^l \zmcD \lambda|_{{\zg}}]
   \le C_1 x^n
 \,,
\ee
where $\mathring \nabla$ denotes the covariant derivative operator of the metric ${\zg}$.

We have:

\begin{theorem}
  \label{T31VII17.1}
Under \eq{18VIII15.6}-\eq{18VIII15.8} and \eq{21XI17.E1}-\eq{18VIII15.9}, let the space-dimension $n$ be greater than or equal to four.
There exists $\epsilon_0 > 0$ such that, for all $0<\epsilon<\epsilon_0 <x_0/4$ there exists a metric $g_\epsilon$, also of the form \eq{18VIII15.6}-\eq{18VIII15.8}, such that
\begin{enumerate}
\item
  $0\le R[g_\epsilon]-R[g ] \leq \frac{Cx^n}{|\ln x|^2}$ for some $C$ independent of $\epsilon$;

  \item $g_\epsilon$
coincides with $g$ for $x > 4\epsilon$;

 \item $g_\epsilon$ has a pure monopole-dipole mass aspect function
   $\Theta_\epsilon$ if $(\transversemanifold, \ringh)$ is conformal to the standard sphere,
   and has constant mass aspect function otherwise;

  \item the associated energy-momentum satisfies
  \begin{equation}\label{10VII17.2b}
   \left\{
     \begin{array}{ll}
       \lim_{\epsilon \to 0} m^\epsilon_0 = m_0\,,
\
 m^\epsilon_i = m_i \,,
& \hbox{if $(\transversemanifold, \ringh)$ is conformal to the round $ {\mathbb S}^{n-1}$;} \\
       \lim_{\epsilon \to 0} m^\epsilon = m
 \,, & \hbox{otherwise.}
     \end{array}
   \right.
\ee
\end{enumerate}
\end{theorem}

\begin{Remark}
  \label{R4XII17.1}
  {\rm
  Note that the decay rate $0\le R[g_\epsilon]-R[g ]= O(\frac{x^n}{|\ln x|^2})$ preserves the integrability condition \eqref{21XI17.E1}.
  }
\end{Remark}

\begin{remark}
{\rm
In dimensions $n \geq 6$, the assumption \eqref{18VIII15.9} can be weakened to
%
\[
 \sum_{0 \leq l \leq 2} |\znabla^l \lambda|_{{\zg}}
   \le C_1 x^n
 \,.
\]
(This can be achieved by using \eqref{28VII17.8} instead of Corollary \ref{Cor:3.10Extended} in the proof; see  Remark \ref{Rem:3.10App}.)
We suspect that this remains true in dimensions $n = 4,5$ but have not attempted to address this.
}
\end{remark}

\noindent{\sc Proof of the Corollary~\ref{C10VIII17.1}:}
We apply Theorem~\ref{T31VII17.1} to the metric obtained by applying to $g$ an isometry of hyperbolic space which maps $(m_0,m_1,\ldots,m_n)$ to  $(m,0,\ldots,0)$, with $m=\pm\sqrt{\ellm _0^2- \sum_{i \geq 1} m_i^2}$. Here the negative sign of $m$ has to be chosen if the original energy-momentum vector was past pointing, positive otherwise.
\qed

Some comments on the proof of Theorem~\ref{T31VII17.1} might be useful. In Step 1  one perturbs the metric $g$ to another metric $\hat g$, which satisfies \eqref{28VII17.4} and whose mass aspect function is purely monopole-dipole, in a manner that the scalar curvature is perturbed in a controlled way.   The metric $\hat g$ obtained in our argument  agrees up to terms which are linear in $\psi$ and $\Phi$ with the metric obtained by first doing a change of variables as in Lemma \ref{Lem:UChangeVar} (so that in the new coordinate system, the metric $g$ satisfies  \eqref{28VII17.9xx}-\eqref{28VII17.9xA}), and then performing a `suitable' truncation to bring the asymptotic behavior back to \eqref{28VII17.4}. It also contains a term which is quadratic in $\Phi$, which needs special care in low dimensions but plays no role in dimensions $n \geq 5$.

As such, the metric $\hat g$ depends on $\epsilon$. In particular, it satisfies \eqref{28VII17.4} with an implicit $\epsilon$-dependent constant for the error terms which deteriorates as $\epsilon \rightarrow 0$. In controlling the scalar curvature, we need $\epsilon$-independent estimates  and, to this end, $\hat g$ needs to be treated as a perturbation of $\mathring{g}$ at order $O(\frac{x^{n-2}}{|\ln x|})$, rather than $O(x^n)$ if the implicit constant in \eqref{28VII17.4} were $\epsilon$-independent. This can be taken care of in dimensions $n \geq 5$ by arranging faster decay in the mixed components $\hat g^{xA}$, after which it is sufficient to work with a perturbation of order $O(x^{n-2})$.

This does not work when $n=4$, but
in this  dimension the logarithmic gain from $O(x^{n-2})$ to $O(\frac{x^{n-2}}{|\ln x|})$, together with the introduction of the quadratic correction term, lead to an error estimate of order $O(\frac{x^{2n-4}}{|\ln x|^2}) = O(\frac{x^{4}}{|\ln x|^2})$ for the scalar curvature, which suffices to take care of the issue. In dimension $n = 3$, the above procedure produces an error of order $O(\frac{x^2}{|\ln x|^2})$, which is too big to be handled by our methods.

\medskip

\noindent{\sc Proof of Theorem~\ref{T31VII17.1}:}
 Let $\Theta$ be the  (standard)  mass aspect function of $g$ in the given asymptotic coordinate system $(x,x^A)$ where \eq{18VIII15.6}-\eq{18VIII15.8} holds.

\medskip
\noindent{\it Step 1.} We will deform the metric $g$ in the asymptotic region to a metric $\hat g$ such that $\hat g$ satisfies   \eqref{28VII17.4}, its mass aspect function $\hat\Theta$ is purely monopole-dipole, and that $R[\hat g] - R[g] = O(\frac{x^n}{|\ln x|^2})$.

Assume that $\psi$ is a smooth function on $\transversemanifold$ such that
\begin{equation}
\langle \psi \rangle = 0
	\,,
\label{31VII17.2}
\end{equation}
where here and below
\[
\langle f \rangle = \frac{1}{\mu_{\ringh}(\transversemanifold)}\int_{\transversemanifold} f\,d\mu_{\ringh}
\]
denotes the average of a function $f$ over $\transversemanifold$ with respect to the measure $d\mu_{{\ringh}{}} $ associated with $\mathring  h_{AB}$.
Thanks to \eqref{31VII17.2}, there exists a function $\Phi: \transversemanifold \rightarrow \mathbb{R}$ such that
\[
\bzmcD _A \bzmcD ^A \Phi = \psi.
\]

For $\epsilon>0$ small we will denote by $\varphi_\epsilon \in C^\infty(\mathbb{R})$ a cut-off function satisfying
\begin{equation}\label{27XII17.1+}
  \varphi_\epsilon=\left\{
     \begin{array}{ll}
       1, & \hbox{$0<x<\epsilon^2$,} \\
       0, & \hbox{$x>\epsilon$,}
     \end{array}
   \right.
\end{equation}
as well as
\begin{equation}\label{27XII17.11+}
 \int_0^\infty \varphi_\epsilon'(x)\,x^{n-2}\,dx =0
	\,,
\end{equation}
together with
\[
|\varphi_\epsilon(x)| \leq C \text{ and } |\varphi_\epsilon'(x)| \leq \frac{C}{x|\ln x|}
 \,,
\]
for some constant $C$ independent of $\epsilon$. See Appendix~\ref{App27XII17.1} for existence of such functions.

We define a new metric $\hat g$ using the formulae
\begin{align*}
T
	&= \lambda - \tr_{\mathring{g}}(\lambda)\mathring{g}
	\,,\\
\hat T^{xx}
	&= T^{xx}
	\,,\\
\hat T^{xA}
	&= T^{xA} - \frac{1}{2} (n-1)\varphi_\epsilon' \, x^n\,\bzmcD ^A \Phi
		\,,\\
\hat T^{AB}
	&=  T^{AB} + \varphi_\epsilon' \,x^{n+1}\,\psi\,\ringh^{AB}
		+\frac{(n-1)^2}{4} x^{2n-2} (\varphi_\epsilon')^2 \,|\bzmcD \Phi|_{{\mathring h}}^2\,\ringh^{AB}\\
		&\qquad\qquad - k \varphi_\epsilon\,x^{n+2}\,\psi\,\ringh^{AB}
			- \frac{1}{n-1}\varphi_\epsilon\,x^{n+2}\,\bzmcD ^C \bzmcD _C \psi\,\ringh^{AB}
		\,,\\
\hat \lambda
	&= \hat T - \frac{1}{n-1}\tr_{\mathring{g}}(\hat T) \mathring{ g}
	\,,\\
\hat g
	&= \mathring{g} + \hat \lambda\,.
\end{align*}
It should be  clear that $\hat g \equiv g$ in the region $\{x \ge \epsilon\}$.

In the region $\{0<x<\epsilon\}$ let
\begin{align}\label{7XII17.1}
 \Xi(x) & = \int_{0}^x \varphi_\epsilon'(s)\,s^{n-2}\,ds
  \equiv \int_{\epsilon^2}^x \varphi_\epsilon'(s)\,s^{n-2}\,ds
  = {O(x^{n-2})}{|\ln \epsilon|^{-1}}
   \,,\\
  X^A(x^C)
  	&= -\frac{1}{2}(n-1)\,\zmcD^A\Phi(x^C)
	\,.
\end{align}
Note that $\Xi $ vanishes for $x<\epsilon^2$ by \eqref{27XII17.1+}, as well as for $x>\epsilon$ by \eqref{27XII17.11+}.

We identify the initial coordinates $(x,x^A)$ for $g$ and the new coordinates $(x,y^A)$ for $\hat g$, as constructed in Corollary~\ref{Cor:3.10Extended}, using
$$
 (x,x^A)\mapsto (x,y^A=x^A)
$$
(thus, \emph{not} $(x,x^A)\mapsto (x,y^A(x^B))$; in other words, we first do the coordinate transformation \eqref{28XII17.1}, and then compare the metric $\hat g$ at a point $(x,y^A=x^A)$ with the metric $g$ at a point $(x,x^A)$, keeping in mind Remark~\ref{R27VIII17.1}).
By Corollary~\ref{Cor:3.10Extended},
applied to $\hat g$, with
$$
 \Upsilon_2(x) = O\big(\frac{x^n}{|\ln x|}\big)
$$
(as defined in \eqref{Eq:18XII17.2}),
and \eqref{28VII17.8} applied to $g$,
we have for $x < \epsilon$
\begin{align}
R[\hat g] - R[g]
	&= x^{n+1}\partial_{x} \Big[
			x^{-1} \partial_{x} \left( x^{-n} (\hat T^{xx} - T^{xx})  \right)
			+ x^{-n-2} \ringh_{AB} (\hat T^{AB} - T^{AB})
 \nonumber
\\
			&\qquad\qquad + 2 x^{-n-1} \bzmcD _A(\hat T^{xA} - T^{xA})\Big]
 \nonumber
\\
			&\qquad\qquad
    {-\frac{(n-1)^3}{4} x^{n+1}\partial_x \Big[x^{n-4} (\varphi_\epsilon')^2 \,|\bzmcD \Phi|_{{\mathring h}}^2\Big] }
    \nonumber
\\
			&\qquad\qquad - k(n-1) x^{n+1}\partial_{x}(x^{-n} (\hat T^{xx}
 - T^{xx})) -
     \frac{(n-2)k}{2} \ringh_{AB} (\hat T^{AB} - T^{AB})
    \nonumber
\\
		&\qquad\qquad  - nkx \bzmcD _A (\hat T^{xA} - T^{xA})
 +
 \bzmcD _A \bzmcD _B (\hat T^{AB} - T^{AB})
 \nonumber
\\
		&\qquad\qquad +O\Big(\max(\frac{x^{n+2}}{|\ln x|}, \frac{x^{2n-4}}{|\ln x|^2})\Big)
				\nonumber
\\
	&= x^{n+1}\partial_{x} \Big[
		  (n-1)
  \big\{
   \underbrace{x^{-1}\,\varphi_\epsilon' \,\psi}_{(a)}
		  	+\underbrace{
 \frac{(n-1)^2}{4} x^{n-4} (\varphi_\epsilon')^2 \,|\bzmcD \Phi|_{{\mathring h}}^2
  }_{(b)}
   \nonumber
\\
			&\qquad\qquad\qquad
		  	- \underbrace{
 k   \varphi_\epsilon\,\psi
  }_{(c)}
				-
 \underbrace{
  \frac{1}{n-1}\varphi_\epsilon\,\bzmcD _A \bzmcD ^A \psi \big\}
   }_{(d)}
    \nonumber
\\
			&\qquad\qquad\qquad \underbrace{-(n-1) x^{-1} \varphi_\epsilon' \,\psi
 }_{(a)}
  \Big]
   \nonumber
\\
			&\qquad\qquad   \underbrace{-\frac{(n-1)^3}{4} x^{n+1}\partial_x \Big[x^{n-4} (\varphi_\epsilon')^2 \,|\bzmcD \Phi|_{{\mathring h}}^2\Big] }_{(b)}
\\
			&\qquad\qquad -
 \underbrace{
  \frac{k(n-1)(n-2)}{2}x^{n+1} \varphi_\epsilon'\,\psi
   }_{(c)}
   \nonumber
\\
			&\qquad\qquad +
 \underbrace{
  \frac{kn(n-1)}{2}x^{n+1} \varphi_\epsilon'\,\psi
   }_{(c)}
				+
 \underbrace{
  \varphi_\epsilon' \,x^{n+1}\,\bzmcD _A \bzmcD ^A \psi
   }_{(d)}
				+ O\Big(\max({x^{n+2}}, \frac{x^{2n-4}}{|\ln x|^2})\Big)
				\nonumber
\\
	&= O\Big(\max({x^{n+2}}, \frac{x^{2n-4}}{|\ln x|^2})\Big)
	\,,
 \label{27XII17.7}
\end{align}
where the groups marked $(a)$, $(b)$, etc., add to zero, and
where the constant in the big $O$ term does not depend on $\epsilon$.

Observe that the metric $\hat g$ satisfies \eqref{28VII17.4} and
\[
 \sum_{0 \leq l \leq 2} \big(|\znabla^l \hat \lambda|_{{\zg}} + x^{-1} |\znabla^l \zmcD\hat \lambda|_{{\zg}} \big)
   \le C(\epsilon)\,x^n
 \,.
\]
In particular, $\hat g$ has a well-defined mass
and its (standard) mass aspect function $\hat\Theta$, as reexpressed in \eqref{28VII17.MAX} in terms of $\hat T$, reads
\begin{align*}
\hat \Theta
	&= \Theta - \frac{1}{n} \ringh_{AB}(\stackrel{(n+2)}{\hat T}{}^{AB} - \stackrel{(n+2)}{T}{}^{AB}) - \frac{2}{n}\bzmcD _A (\stackrel{(n+1)}{\hat T}{}^{xA} - \stackrel{(n+1)}{T}{}^{xA}) - \frac{2}{n} (\stackrel{(n+2)}{\hat T}{}^{xx} - \stackrel{(n+2)}{T}{}^{xx})\\
	&= \Theta + \frac{(n-1)}{n}\Big[k \psi + \frac{1}{n-1} \bzmcD _A \bzmcD ^A \psi\Big]	\\
	&= \Theta + \frac{1}{n} \bzmcD _A \bzmcD ^A \psi + \frac{k(n-1)}{n}\psi.
\end{align*}

We now proceed to choose $\psi$. Consider first the case when $\transversemanifold$ is \emph{not} the standard sphere. By a result of Lichnerowicz  and of Obata~\cite{Lichnerowicz58,Obata62} (compare~\cite{Kuehnel}), the first eigenvalue of the Laplacian is strictly larger than $n - 1$. Therefore, there exists $\psi$ such that
\begin{equation}
\bzmcD _A\,\bzmcD ^A \psi + k(n-1)\psi = n(\langle \Theta \rangle - \Theta)
	\,.\label{31VII17.1}
\end{equation}
(When $k \neq 0$, $\psi$ exists as $\bzmcD _A\,\bzmcD ^A \psi + k(n-1)$ is injective. When $k = 0$, $\psi$ exists as the right-hand side of \eqref{31VII17.1} has zero average.) Furthermore, if $k \neq 0$, we see by integrating both sides of \eqref{31VII17.1} that $\langle \psi \rangle = 0$, i.e. \eqref{31VII17.2} is satisfied. If $k = 0$, $\psi$ is determined up to an additive constant, which can be arranged so that \eqref{31VII17.2} is satisfied. In any event, we obtain a solution of \eqref{31VII17.1} which also satisfies \eqref{31VII17.2}. This leads to
\[
\hat \Theta = \langle\Theta\rangle
	\,.
\]

Consider next the case when $\transversemanifold$ is the standard sphere (in which case $k = 1$). It is well known that $n-1$ is the first eigenvalue of the Laplacian. Let $\Theta_0 = \langle \Theta\rangle$ and $\Theta_1$ be respectively the orthogonal projection of $\Theta$ onto the zeroth and first eigenspaces of the Laplacian. Then there exists a solution of
\begin{equation}
\bzmcD _A\,\bzmcD ^A \psi + k(n-1)\psi = n(\Theta_0 + \Theta_1 - \Theta)
	\,.\label{31VII17.X}
\end{equation}
Integrating both sides of the above equation, we see that  \eqref{31VII17.2} is also satisfied. We thus obtain
\[
\hat\Theta = \Theta_0 + \Theta_1
	\,,
\]
which concludes Step 1.

\medskip
\noindent{\it Step 2.} We proceed to deform $\hat g$ to the desired metric.

Let $\tvarphiep(x) := \tilde\varphi(\frac{x - 2\epsilon}{2\epsilon})$ {where, the cut-off function $\tilde\varphi\in C^\infty(\R)$ equals to one in $(-\infty,0]$ and vanishes on $[1,\infty)$. One can, and it is convenient to, assume that
$$
 |\tilde \varphi'|^2 \leq C\,\tilde\varphi
  \,.
$$
From Step 1, there exists some constant $C_1$ independent of $\epsilon$ such that
\begin{equation}\label{27XII17.5}
R[\hat g] - R[g] \geq - C_1\,\tvarphiep(x)\, \frac{x^{n}}{|\ln x|^2}
	\,.
\end{equation}
Here we have used that $n \geq 4$.

Consider
\begin{equation}\label{27XI17.1}
\check g  = \hat g + \frac{1}{n-1}\xi(x)\,\mathrm{d}x^2
\end{equation}
where, for some $C_* > 0$ to be specified,
\[
\xi(x) = - C_*\,x^{n-2}\int_{4\epsilon}^{x} \frac{\tvarphiep(s)}{s\,|\ln s|^2} \,ds
 =  C_*\,x^{n-2}\int_{x}^{\infty} \frac{\tvarphiep(s)}{s\,|\ln s|^2} \,ds
	\,.
\]
Note that, as $\varphi$ is non-increasing and non-negative,
\begin{equation}
\Big|\int_{4\epsilon}^{x} \frac{\tvarphiep(s)}{s\,|\ln s|^2} \,ds\Big| \leq \tvarphiep(x)
   |\ln(4 \epsilon)|^{-1}
  \qquad \text{ for } x < 4\epsilon
	\,.\label{Eq:01XII17.1}
\end{equation}
Thus $\xi(x)$ vanishes for $x\ge 4\epsilon$, while for $0<x<4 \epsilon<1/2 $ we have
\begin{equation}\label{27XII17.8}
 0\le  \xi (x) \le C x^{n-2}\tvarphiep(x)|\ln \epsilon|^{-1}  = O (x^{n-2}|\ln \epsilon|^{-1}) = O (x^{n-2}|\ln x|^{-1})
	\,.
\end{equation}

The tensor $\check T$ corresponding to $\check g$ is
\begin{align*}
\check T^{xx}
	&= \hat T^{xx}
	\,,\\
\check T^{xA}
	&= \hat T^{xA}
	\,,\\
\check T^{AB}
	&= \hat T^{AB} - \frac{1}{n-1} x^{{4}}\,\big(1 - \frac{k}{4}x^2\big)^{{-2}}\,\xi(x)\,\ringh^{AB}
	\,.
\end{align*}

For $x \leq 4\epsilon$, after inspecting the calculations of Corollary \ref{Cor:3.10Extended} to determine $R[\check g] - R[\hat g]$, one finds
\begin{align}
 \nn
 R[\check g] - R[g]=
 &
 { R[\check g] - R[\hat g] + R[\hat g] - R[g]}
\\
 \nn
	\geq
 &
  -x^{n+1}\partial_{x}\Big\{x^{-(n-2)} \xi(x)\Big\}
		\underbrace{-{\frac{k(n-2)}2 x^{{4}} \xi}
    }_{\ge -C_2 {C_*}\tvarphiep(x)\, \frac{x^{{n+2}}}{|\ln x| }}
\\
 \nn
 		&
  - C_2\,\tvarphiep(x)\max({x^{n+2}}, \frac{x^{2n-4}}{|\ln x|^{2}})
\\
		&
 - C_2(x^6\,\omega_1(x)  + x^{4} \omega_2(x)^2)
 { + R[\hat g] - R[g]}
		\,,
 \label{27XII17.6}
\end{align}
where $C_2$ is independent of $\epsilon$ and
\[
\omega_\ell(x) = \sum_{0 \leq j \leq \ell} x^j\,|\partial_x^j \xi(x)|, \qquad \ell = 1, 2
	\,.
\]
Using \eqref{Eq:01XII17.1} and the fact that $|\tilde\varphi'|^2 \leq C\tilde\varphi$, we can bound
\[
x^6\,\omega_1(x)  + x^{4} \omega_2(x)^2 \leq C_3\,C_*(C_* + 1) \tvarphiep(x)\,\frac{x^{n +4}}{|\ln x|}
\]
for some $C_3$ independent of $\epsilon$. In view of \eqref{27XII17.5} it should be clear that a constant $C_*$ can be chosen such that,  for all sufficiently small $\epsilon$, there holds
\begin{align*}
R[\check g] - R[g]
	&\geq \frac{C_*}{2}\,\tvarphiep(x)\,\frac{x^n}{|\ln x|^2}
		\,,
\end{align*}
which implies that
\[
R[\check g] \geq R[g]
	\,.
\]
It is also clear that $R[\check g] \leq R[g] +\tilde \varphi _{2\epsilon}(x)\,O(\frac{x^n}{|\ln x|^2})$.

The metric $\check g$ is readily seen to be of the form \eqref{28VII17.4}, and so, by Corollary \ref{Cor:29XII17-C1}, of the form \eq{18VIII15.6}-\eq{18VIII15.8} after a suitable coordinate transformation at infinity. The mass aspect function $\check\Theta$ of $\check g$ is found to be
\begin{align*}
\check\Theta
	&= \hat \Theta +  \frac{C_*}{n}\,\int_0^{4\epsilon} \frac{\tvarphiep(s)}{s|\ln s|^2}\,ds\\
	&= \hat \Theta +  O(|\ln \epsilon|^{-1})
		\,.
\end{align*}
This concludes the proof.
\qedskip

\begin{remark}
 \label{R4V18.1}
 If the metric $g$ in Theorem \ref{T31VII17.1} is $C^k$--conformally compactifiable, $3 \leq k \leq \infty$, then the metrics $g_\epsilon$ constructed above are $C^{\min(k,n+1)}$--conformally compactifiable;
in fact, $C^{n+1|k-(n+1)}$--conformally compactifiable for $k> n+1$.
When $n \geq 5$, the proof can be slightly modified to obtain metrics $g_\epsilon$ which are $C^k$--conformally compactifiable.
\end{remark}

\begin{proof} After Step 1 of the proof, the metric $\hat g$ is $C^k$--conformally compactifiable. In Step 2, note that $\xi$ is a multiple of $\frac{x^{n-2}}{\ln x}$ near $x = 0$, and so $g$ appears to be $C^{\min(k,n)}$--conformally compactifiable. However, we have
\[
(g_\epsilon)_{xx} = g_{xx} + \frac{1}{n-1}\xi + x^{n-2} \zeta(x^A)
\]
where $\zeta$ is a smooth function on $\transversemanifold$.
We can thus pass from the original coordinate system $(x,x^A) $ to new coordinates $(y,y^A)$ by setting $y^A = x^A$, while $y$ is obtained by integrating
$$
 \mathrm{d}y^2  =  \left(1+ \frac{x^2}{n-1}\xi(x)\right)\,\mathrm{d}x^2
 \,.
$$
In this coordinate system $  g_\epsilon $ is   $C^{\min(k,n+1)}$--conformally compactifiable.

When $n \geq 5$, one can modify the proof of Theorem~\ref{T31VII17.1} to obtain a $C^{k}$--conformally compactifiable metric by letting instead
\[
\xi(x) =  C_*\,x^{n-2}\int_{x}^{\infty} \tvarphiep(s)\,s \,ds
	\,.
\]
This is because, in place of \eqref{27XII17.5}, we have in these dimensions the estimate $R[\hat g] - R[g] \geq - C_1\,\tvarphiep(x)\, x^{n+2}$.
\qedskip
\end{proof}

\section{Miscellaneous}
 \label{s28XII17.1}

In this section we point-out some miscellaneous results concerning the mass aspect function $\Theta$.
We start by  noting that $\Theta$ can always be ``pushed-up'' by an arbitrary amount.

\begin{theorem}
  \label{T28VIII17.1}
  Under \eq{18VIII15.6}-\eq{18VIII15.8} and \eq{21XI17.E1}-\eq{18VIII15.9}, let the space-dimension $n$ be greater than or equal to three. Let $\Theta:\transversemanifold \rightarrow \mathbb{R}$ be the mass aspect function of $g$ and let
  $$
   \eta: \transversemanifold \rightarrow \mathbb{R}
  $$
  be a smooth non-negative function.  There exists $\epsilon_0 > 0$ such that, for all $0<\epsilon<\epsilon_0 <x_0/4$ there exists a metric $g_\epsilon$, also of the form \eq{18VIII15.6}-\eq{18VIII15.8}
  (after possibly a coordinate transformation), such that

\begin{enumerate}
\item
  $
   0\le R[g_\epsilon]-R[g ] =O(x^{n+2})
   \,,
   $%

  \item $g_\epsilon$
coincides with $g$ for $x > 4\epsilon$,

 \item $g_\epsilon$ has mass aspect function $\Theta + \eta + c_\epsilon$ for some non-negative constant $c_\epsilon$ which tends to zero as $\epsilon \rightarrow 0$.

\end{enumerate}
\end{theorem}

\begin{proof}
Let $0 \leq \varphi \in C^\infty(\mathbb{R})$ be a non-increasing cut-off function which equals $1$ in $(0,1)$ and vanishes identically in $(2,\infty)$. Fix some small $\epsilon > 0$, and let
\[
\varphi_\epsilon(x) = \varphi(\epsilon^{-1} x).
\]

Consider
\begin{equation}\label{27XI17.1a}
\hat g  = g + \frac{n}{n-1}\,x^{n-2}\,\varphi_\epsilon(x)\,\eta(x^A)\,\mathrm{d}x^2
	\,.
\end{equation}
For $ x \geq 2\epsilon$, we have $\hat g \equiv g$. By \eqref{6IX15.1}, we have for $x \leq 2\epsilon$ that
\begin{align*}
R[\hat g] - R[g]
	&\geq - nx^{n+1}\partial_{x}\Big\{ \varphi_\epsilon(x)\,\eta(x^A)\Big\} + O(x^{n+2})\\
	&= - nx^{n+1}\underbrace{\varphi_\epsilon'(x)}_{\leq 0}\,\eta(x^A) + O(x^{n+2})\\
	&\geq O(x^{n+2})
		\,,
\end{align*}
where the error term is larger than $-C\,x^{n+2}$ for some $C$ independent of $\epsilon$.

The metric $\hat g$ is of the form \eqref{28VII17.4}, and so of the form \eq{18VIII15.6}-\eq{18VIII15.8} after a suitable coordinate transformation at infinity as in Corollary \ref{Cor:29XII17-C1}. The mass aspect function $\hat\Theta$ of $\hat g$ is related to the mass aspect function $\Theta$ by
\begin{align*}
\hat\Theta
	&= \Theta + \eta
		\,.
\end{align*}
We now follow Step 2 in the proof of Theorem \ref{T31VII17.1} to deform $\hat g$ (in the asymptotic region) to a metric $\check g$ such that $\check g \equiv g$ for $x \geq 4\epsilon$, $0 \leq R[\check g] - R[g] \leq O(x^{n+2})$, and the mass aspect function $\check\Theta$ can be written in the form $\check\Theta = \hat \Theta + c_\epsilon$ for some constant $c_\epsilon = O(\epsilon^{2})$.
\end{proof}
\qedskip

\begin{corollary}
  \label{Cor28VIII17.1}
  Under \eq{18VIII15.6}-\eq{18VIII15.8} and \eq{21XI17.E1}-\eq{18VIII15.9}, let the space-dimension $n$ be greater than or equal to three. Assume that $(\transversemanifold,\ringh)$ is conformal to the standard sphere and let the energy-momentum covector of $(M,g)$ be $(m_0, m_1, \ldots, m_n)$. Let $(\tilde m_0, \tilde m_1, \ldots, \tilde m_n)$ be an energy-momentum covector which lies to the chronological future of $(m_0, m_1, \ldots, m_n)$, i.e. $\tilde m_0 > m_0$ and $(\tilde m_0 - m_0)^2 - \sum_{i \geq 1} (\tilde m_i - m_i)^2 > 0$. There exists $\epsilon_0 > 0$ such that, for all $0<\epsilon<\epsilon_0 <x_0/4$ there exists a metric $g_\epsilon$, also of the form \eq{18VIII15.6}-\eq{18VIII15.8}
    (after possibly a coordinate transformation), such that

\begin{enumerate}
\item
  $
   0\le R[g_\epsilon]-R[g ] =O(x^{n+2})
   \,,
   $%

  \item $g_\epsilon$
coincides with $g$ for $x > 4\epsilon$,

 \item $g_\epsilon$ has an energy-momentum covector $(m_0^\epsilon, m_1^\epsilon, \ldots, m_n^\epsilon)$ such that $m_\mu^\epsilon \rightarrow \tilde m_\mu$ as $\epsilon \rightarrow 0$.
\end{enumerate}
\end{corollary}

\begin{remark}
 \label{R28XII17.1}
{\rm
In particular, if $(m_\mu)$ is timelike past-pointing, the above produces metrics $g_\epsilon$ with $(m_\mu^\epsilon)$ as close to zero as desired.
}
\end{remark}

\begin{proof} We view $\transversemanifold \approx \mathbb{S}^{n-1}$
 as being standardly embedded in $\mathbb{R}^n$ so that the first eigenfunctions of the Laplacian on $\mathbb{S}^{n-1}$ are the coordinate functions $x_i \equiv x^i$ of $\mathbb{R}^n$.
 Fix some non-negative smooth function $\eta: \mathbb{S}^{n-1} \rightarrow [0,\infty)$ for the moment, and let $\check g$ be the metric obtained in the proof of Theorem \ref{T28VIII17.1}. We proceed to compute the energy-momentum covector $(\check m_0, \check m_1, \ldots, \check m_n)$ of $\check g$. We have
\begin{align*}
\check m_0
	&= m_0 + c_n\,\int_{\mathbb{S}^{n-1}} \eta\,d\mu_{\ringh} + O(\epsilon^{2})
		\,,\\
\check m_i
	&= m_i +  c_n\,\int_{\mathbb{S}^{n-1}} \eta\,x_i\,d\mu_{\ringh} + O(\epsilon^{2})
		\,,\enspace i = 1, \ldots, n
		\,.
\end{align*}
Thus, to conclude the argument, it suffices to show that $\eta$ can be chosen such that
\begin{align*}
\int_{\mathbb{S}^{n-1}} \eta\,d\mu_{\ringh}
	&= \frac{1}{c_n} (\tilde m_0 - m_0) =: v_0
		\,,\\
\int_{\mathbb{S}^{n-1}} \eta\,x_i d\mu_{\ringh}
	&= \frac{1}{c_n} (\tilde m_i - m_i) =: v_i, \enspace i = 2, \ldots, n
 \,.
\end{align*}
To this end, we may assume without loss of generality (after a suitable rotation of coordinate axes of $\mathbb{R}^n$) that $v_1 \geq 0$ and $v_2 = \ldots = v_n = 0$. Note that by assumption, the covector $(v_0, v_1, \ldots, v_n)$ is timelike and so $v_0 > v_1 \geq 0$. Select a function $a \in C_c^\infty(v_1\,v_0^{-1},1)$ such that $a \geq 0$ and $a \neq 0$. Then $\eta$ can be chosen as
\[
\eta = \alpha + \beta\,a(x^1)
\]
where
\begin{align*}
\alpha
	&= \Big(\mu_{\ringh}(\mathbb{S}^{n-1}) \int_{\mathbb{S}^{n-1}} a(x^1)\,x^1\,d\mu_{\ringh}\Big)^{-1} \Big\{v_0 \int_{\mathbb{S}^{n-1}} a(x^1)\,x^1\,d\mu_{\ringh} - v_1 \int_{\mathbb{S}^{n-1}} a(x^1)\,d\mu_{\ringh}\Big\},\\
\beta
	&=  \Big( \int_{\mathbb{S}^{n-1}} a(x^1)\,x^1\,d\mu_{\ringh}\Big)^{-1} \,v_1
		\,.
\end{align*}
Clearly $\beta \geq 0$ and, thanks to the requirement that the support of $a$ is contained in $(v_1\,v_0^{-1},1)$, $\alpha \geq 0$. The conclusion is readily seen.
\end{proof}
\qedskip

\section{Applications}
 \label{s29X17.1}

Let $(M^n,g)$ be an asymptotically locally hyperbolic (ALH) manifold as defined above.
We assume that  $x^{-1}(R[g]+ n(n-1)) $  is in $L^1$  and that the decay hypotheses needed for the deformation results above hold.
We consider the case where $M$ is complete with compact  boundary
 satisfying $H< (n-1)$, where the mean extrinsic curvature $H$ is calculated with respect to the inner pointing normal. The models we have in mind in Theorem~\ref{pmass1} are the higher dimensional black holes discussed, in, e.g.\ Birmingham's paper~\cite{Birmingham}, in the cases $k = 0$ and $k =1$.  More specifically, we are interested in the cases where $N^{n-1}$ is a torus  or a (nontrivial) quotient of a sphere.

\medskip

\noindent{\sc Proof of the torus case of Theorem~\ref{pmass1}:}
Suppose $m < 0$. Hence,  deforming the metric slightly near the conformal boundary if necessary, we may assume  by Theorem~\ref{T31VII17.1} that the mass aspect function is negative.

We claim that,  for $r$ large enough (i.e., $x$ sufficiently close to $0$), $r = r_1$, say, the level surface $N_1 = \{r_1\} \times N$ has mean curvature  $H_1 > n-1$    calculated with respect to the normal  $\nu$  pointing towards conformal infinity.  To this end, we recall {\eqref{18VIII15.6}-\eqref{18VIII15.8}} (note that $k = 0$):
\[
  g = x^{-2} \Big(\mathrm{d}x^2  +  {\ringh}{}     + x^{n} \mu \Big)+o(x^{n-2}) \mathrm{d}x^i \mathrm{d}x^j
\,.
\]
The one-form dual to the normal $\nu$ to $\{r = r_1\}$ is
$$-\frac{1}{g(dx,dx)^{1/2}} dx = -\frac{1 + o(x^{n})}{x} dx
 \,.
$$
Thus,
\begin{align*}
H_1
	&= g^{AB}\,\nabla_{A} \nu_B \\
	&= x(1 + o(x^{n}))(\ringh^{AB} - x^n\,\ringh^{AC}\ringh^{BD} \mu_{CD} + o(x^n))\,\Gamma_{AB}^1\\
	&= \frac{1}{2} x^3(\ringh^{AB} - x^n\,\ringh^{AC}\ringh^{BD} \mu_{CD} + o(x^n))(-\partial_x g_{AB} + o(x^{n-2}))\\
	&= \frac{1}{2} (\ringh^{AB} - x^n\,\ringh^{AC}\ringh^{BD} \mu_{CD} + o(x^n))(2 \ringh_{AB} - (n-2)x^{n}\mu_{AB} + o(x^{n}))\\
	&= (n-1) - \frac{1}{2}n\,x^n \tr_{\ringh}\mu + o(x^n)
		\,.
\end{align*}
As the mass aspect function $\tr_{\ringh}\mu$ is negative, the claim follows.

For this proof we find it convenient to use standard existence results for marginally outer trapped surfaces (MOTSs); see \cite{AEM} and references therein.  To this end,  we introduce a second fundamental form: $K = - g$, and consider the initial data set $(M,g,K)$.   Observe that the scalar curvature condition implies that the dominant energy condition,
$\mu \ge |J|$, holds.

Now consider the compact body $W = [r_0, r_1]  \times  N$, with boundary $N_0 \cup N_1$.  For the null expansion $\theta _0$ of $N_0$, with respect to the normal pointing into $W$, we have
$\theta _0 = H_0 + {\rm tr}_{N_0} K < (n-1) - (n-1) = 0$.   For the null expansion
 $\theta _1$ of $N_1$, with respect to the normal pointing out of $W$, we have $\theta _1 = H_1 + {\rm tr}_{N_1} K > (n-1) - (n-1) =0$.
 Under these barrier conditions there exists an `outermost' MOTS ${\Sigma}$ in the interior of $W$; that is, ${\Sigma}$ encloses $N_0$, and there is no MOTS, or, more generally, weakly outer trapped surface ($\theta  \le 0$), enclosing
${\Sigma}$ (see \cite[Theorem 4.6]{AEM} and \cite[Theorem 5.1]{eichmairgah}).

In general, ${\Sigma}$ may have several components.
Using the  product structure of $W$, we obtain a  projection map $P: W \to N_0$,
such that $P\circ j = {\rm id}$, where $j: N_0 \to \hat M$ is inclusion. The map
$f = P \circ i : {\Sigma} \to N_0$, where $i : {\Sigma} \to \hat M$ is inclusion, induces
a map on homology $f_*: H_{n-1}({\Sigma}) \to H_{n-1}(N_0)$.   Using that ${\Sigma}$ is homologous to $N_0$, we compute, $f_*[{\Sigma}] = P_*(i_*[{\Sigma}]) = P_*(j_*[N_0]) = {\rm id}_*[N_0] = N_0 \ne 0$.  It follows that there is a component ${\Sigma}'$ of ${\Sigma}$, for which there is a nonzero degree map from ${\Sigma}'$ to $N_0$.
Hence, by a result of Schoen and Yau  \cite[Corollary~2]{SYsc2}, ${\Sigma}'$ does not carry a metric of positive scalar curvature.  It then follows from Theorem 3.1 in \cite{G} that an outer neighborhood of ${\Sigma}'$ is foliated by MOTSs.  But this contradicts ${\Sigma}$ being outermost.
\qed

We pass now to:

\medskip

\noindent{\sc Proof of the sphere case of Theorem~\ref{pmass1}:}
Here we apply more directly results and arguments from \cite{AnderssonGallowayCai}.
Suppose $m < 0$. Hence,  again by Theorem~\ref{T31VII17.1} we may assume that the mass aspect function of $(M,g)$ is negative. 
By Remark~\ref{R4V18.1} and Proposition~\ref{Prem:02V18-R1}, we may also assume that the terms $o(x^{n-2})dx^idx^j$ in \eqref{18VIII15.6} are actually $o(x^{n-2})dx^Adx^B$, as assumed in \cite{AnderssonGallowayCai}.   Pass to the Riemannian universal cover
$(M',g')$.  We have $M' = [r_0, \infty) \times N'$, where $(N',\mrh')$ is a round sphere covering $(N, \mrh)$, and $N'_0 = \{r_0\} \times N'$ has mean curvature $H < n-1$. Moreover, the mass aspect function will be negative in $(M',g')$.  Then by \cite[Theorem 3.2]{AnderssonGallowayCai}, $g'$ can be deformed to a metric $g''$ on $M'$ such that:
\ben
\item $R(g'') \ge - n(n-1)$, and for some numbers $r_1 < r_2$,
\item $g'' = g'$ (up to homothety) inside $r = r_1$,
\item $g'' =$ hyperbolic metric outside $r = r_2$.
\een

Now using the `translational isometries' of the half space model for hyperbolic space, we obtain the identification space $(\hat M, \hat g)$, which, outside a compact set $K$, is given by
(see \cite[Section 2.3]{AnderssonGallowayCai}),
\beq\label{cusp}
\hat M = \bbR \times T \,,  \quad \hat g = dt^2 + e^{2t} h \,,
\eeq
where $(T,h)$ is a flat torus.  Thus, $(\hat M, \hat g)$ is just
a standard hyperbolic cusp outside the compact set $K$,
with scalar curvature $S[\hat g] \ge -n(n+1)$ everywhere, and with a spherical boundary ${\Sigma}_0$, say, contained in $K$, having mean curvature $H < n-1$.

Fix a  large number $b > 0$, so that $K$ lies in the region $-b < t < b$. Let $W$ be the region of $\hat M$ bounded between the toroidal slices ${\Sigma}_1 = \{-b\} \times T$ and
${\Sigma}_2 = \{b\} \times T$; thus $W$ is compact with boundary components ${\Sigma}_i$, $i =0, 1, 2$.

Now consider the `brane action' $\calB$:  For any compact hypersurface ${\Sigma}$ in $W$ homologous to ${\Sigma}_2$ (equivalently, homologous to ${\Sigma}_0 \cup {\Sigma}_1$),
\beq
{ \calB}({\Sigma}) = A({\Sigma}) - (n-1)V({\Sigma}) \, ,
\eeq
where $A({\Sigma}) =$ area of ${\Sigma}$ and $V({\Sigma}) =$ the volume of the region bounded by ${\Sigma}$ and ${\Sigma}_0 \cup {\Sigma}_1$.  We now minimize $\calB $ among all such hypersurfaces, as in \cite{AnderssonGallowayCai}.  The difference here is the presence of the boundary component ${\Sigma}_0$.  However, it has mean curvature
$< (n-1)$ with respect to the normal pointing into $W$, and, as such, forms an appropriate barrier for the minimization process.  As described in \cite{AnderssonGallowayCai},  using standard regularity results from geometric measure theory, we obtain a smooth compact embedded minimizer $S$ for the brane action, homologous to
${\Sigma}_2$.  From the discussion in \cite[Section 2.3]{AnderssonGallowayCai}, the minimizer can be constructed so as to lie in the region $-b < t < b$, and  hence is contained in the interior of $W$.

Using the `almost product' structure of $W$, there exists a retract of $W$ onto ${\Sigma}_2$.  Arguing as in the torus case (see also~\cite[Section 2.3]{AnderssonGallowayCai}), one finds that there is a nonzero degree map from some component $S'$ of $S$ to the torus  ${\Sigma}_2$. It follows from the result of Schoen and Yau \cite{SYsc2} alluded to above that $S'$ cannot carry a metric of positive scalar curvature. Then, since $S'$ must minimize the brane action in its homology class, Theorem 2.3 in \cite{AnderssonGallowayCai} gives that a neighborhood $U$ of $S'$ splits as a warped product,
\beq\label{warped}
U = (-u_1, u_2) \times S' \qquad \hat g|_U = du^2 + e^{2u} h  \, ,
\eeq
where the induced metric $h$ on $S'$ is flat.  But since $S'$ in fact {\it globally}
minimizes the brane action in its homology class,  this local
warped product structure can be extended to  larger $u$-intervals. Extend the warped product to larger
values of $u_2$ (keeping $u_1$ fixed for the moment).  Using the fact that $S$ is separating, eventually
$S'_2 = \{u_2\} \times S'$ will meet ${\Sigma}_2$ or another component of $S$ (without meeting ${\Sigma}_0$).  However, since $S'_2 \cup (S \setminus S')$ minimizes the brane action, the latter cannot occur: Where they touch, one could remove small disks of radius $\delta$, which contribute a term of order $O(\delta^2)$ to the brane action, and insert a cylinder, which contributes a term of order $O(\delta^3)$ to the brane action, so as to decrease the brane action, thereby contradicting the minimality of ${\calB }(S)$.%
\footnote{Alternatively, as $S'_2 \cup (S \setminus S')$ is minimizing, it is a regular embedded surface and so $S_2'$ cannot touch $S \setminus S'$.}
Hence, $S'_2$ meets ${\Sigma}_2$, and by the maximum principle, they agree.  This implies that $S'$ is homologous to ${\Sigma}_2$, and hence homologous to ${\Sigma}_0 \cup {\Sigma}_1$.

Now continue the warped product \eqref{warped} to more negative $u_1$-values, until at
some such value,  $S'_1 = \{-u_1\} \times S'$ meets ${\Sigma}_0$.  (If
the warped product reached ${\Sigma}_1$ without touching ${\Sigma}_0$, then $S'$ would be homologous to
${\Sigma}_1$, contradicting that it is homologous to ${\Sigma}_0 \cup {\Sigma}_1$.)
But, by a basic mean curvature comparison result, this one-sided tangential intersection is incompatible with the fact that ${\Sigma}_0$ has mean curvature smaller than $(n-1)$ and $S_1'$ has mean curvature equal to $(n-1)$ with respect to its `inward' normal. Hence, we arrive at a contradiction.\qed

\begin{remark}
\label{R10I18.1}
{\rm
 It is interesting to consider the torus case of Theorem~\ref{pmass1} in the context of the Horowitz-Myers AdS soliton \cite{HorowitzMyers}.   The AdS soliton is a globally static spacetime satisfying the vacuum Einstein equations with negative cosmological constant, which has negative mass.  Each time slice has topology $\bbR^2 \times T^{n-2}$.  Removing an open radial disk from the $\bbR^2$ factor, one obtains an ALH manifold $M = [r, \infty) \times T^{n-1}$, which, under appropriate scalings, satisfies all the assumptions of the  torus case of Theorem~\ref{pmass1} (in dimensions $4 \le n \le 7$), except for the mean curvature condition.  The mean curvature $H(r)$ of the boundary $N_r = \{r\} \times T^{n-1}$ is always greater than $n-1$, but comes arbitrarily close to this value as $r$ becomes arbitrarily large.  In this sense, one sees that Theorem~\ref{pmass1}, in the torus case, is  essentially sharp.
 \qed
 }
\end{remark}
 
It is perhaps worth noting that the torus case generalizes to the case of a compact flat (i.e. curvature zero) conformal infinity, provided the product assumption in Theorem 1.1 extends to the conformal completion.  This follows from a covering space argument using the fact that any compact flat manifold is finitely covered by a flat torus.

 \appendix

\section{Variations of the metric and scalar curvature}
 \label{app17VIII15.1}

In this appendix we  estimate in detail the error terms arising in our argument.
For this, consider a metric of the form
\begin{equation*}
\hat g = g+q = {\zg} + \lambda + q
\,,
\end{equation*}
where $\lambda$ and $q$ are thought of as being small compared to  $\zg$, in the sense that
\begin{eqnarray*}
|\lambda|_{{\zg}} + |\mathring\nabla\lambda|_{{\zg}}  + |\mathring\nabla\mathring\nabla\lambda|_{{\zg}}  +
|q|_{{\zg}} + |\mathring\nabla q|_{{\zg}}  + |\mathring\nabla\mathring\nabla q|_{{\zg}}
  \le \delta
  \,,
\end{eqnarray*}
 where the error terms are understood in ${\zg}$-norm.
Here the metric ${\zg}$ is considered to be general, not necessarily given by \eq{17VIII15.1}.

Given a metric $g_1$ and a small symmetric tensor $\lambda_1$ we will use the following formulae
\begin{eqnarray}
 R[g_1+\lambdatwo ] &=& R[g_1] + (\bzmcD _{g_1}R[g_1])\lambdatwo  + O(|\lambdatwo |^2_{g_1}) + O(|\nabla_{g_1}\lambdatwo |^2_{g_1})
\nonumber
\\
 &&
    + O(|\lambdatwo |_{g_1}|\nabla_{g_1}\nabla_{g_1}\lambdatwo |_{g_1})
 \,,
\\
\mathrm{Ric}[g_1+\lambdatwo ] &=&\mathrm{Ric}[g_1] + (\bzmcD _{g_1}\mathrm{Ric}[g_1])\lambdatwo
    + O(|\lambdatwo |_{g_1}|\nabla_{g_1}\lambdatwo |_{g_1})
\nonumber
\\
    &&
    + O(|\nabla_{g_1}\lambdatwo |^2_{g_1})
+ O(|\lambdatwo |_{g_1}|\nabla_{g_1}\nabla_{g_1}\lambdatwo |_{g_1})
\,,
\label{Ric_formula}
\\
\bzmcD _{g_1}R[g_1]\lambdatwo  &=& - {\nabla_{g_{1}}}_k {\nabla_{g_{1}}}^k\mathrm{tr}_{g_1}\lambdatwo  + {\nabla_{g_{1}}}^k{\nabla_{g_{1}}}^l {\lambdatwo }_{kl} -\mathrm{Ric}[g_1]^{kl}{\lambdatwo }_{kl}
\,,
\phantom{xxxxx}
\\
(\bzmcD _{g_1}\mathrm{Ric}[g_1] \lambdatwo )_{ij}
&=& {\nabla_{g_1}}^k{\nabla_{g_1}}_{(i} {\lambdatwo }_{j)k}  -\frac{1}{2} {\nabla_{g_{1}}}_k {\nabla_{g_{1}}}^k {\lambdatwo }_{ij}
\nonumber
\\
&&
- \frac{1}{2}{\nabla_{g_1}}_{i}{\nabla_{g_1}}_{j} \mathrm{tr}_{g_1}{\lambdatwo }
\,.
\end{eqnarray}
Moreover, the inverse metric $g^{{\mathrm{inv}}}$ satisfies
\begin{eqnarray*}
 g^{{\mathrm{inv}}}  - {\zg}^{{\mathrm{inv}}} \,=\,  O(|\lambda |_{{\zg}})
\,,\quad
\Gamma^k_{ij} - \mathring \Gamma^k_{ij} \,=\,  O(|\mathring\nabla\lambda |_{{\zg}})
\,.
\end{eqnarray*}
That yields
\begin{eqnarray}
 \nn
R[g+q] &=&
R[g] - \mathring \Delta \mathrm{tr}_{{\zg}} q
+ \mathring \nabla^k \mathring \nabla^l q_{kl}
 -   \zg^{ik} \zg^{jl} R_{ij}[{\zg} + \lambda ] q_{kl}
\nonumber
\\
&&+ O(|q|_{\zg}^2) + O(|\znabla q|_{\zg}^2)+ O(|q|_{\zg} |\znabla\znabla q|_{\zg})
\nonumber
\\
&&
 + O(|q|_{\zg}|\znabla \lambda|_{\zg}) + O(|\znabla q|_{\zg}|\znabla \lambda|_{\zg})+ O(|q|_{\zg}|\znabla\znabla \lambda|_{\zg})
\nonumber
\\
&&
 + O(|\lambda |_{\zg}| q|_{\zg})  + O(|\lambda |_{\zg}|\znabla q|_{\zg}) + O(|\lambda |_{\zg}|\znabla\znabla q|_{\zg})
\\
&=&R[g]+ \bzmcD _{{\zg}}R[{\zg}] q
\nonumber
\\
&&+ O(|q|_{\zg}^2) + O(|\znabla q|_{\zg}^2)+ O(|q|_{\zg} |\znabla\znabla q|_{\zg})
\nonumber
\\
&&
 + O(|q|_{\zg}|\znabla \lambda|_{\zg}) + O(|\znabla q|_{\zg}|\znabla \lambda|_{\zg})+ O(|q|_{\zg}|\znabla\znabla \lambda|_{\zg})
\nonumber
\\
&&
 + O(|\lambda |_{\zg}| q|_{\zg})  + O(|\lambda |_{\zg}|\znabla q|_{\zg}) + O(|\lambda |_{\zg}|\znabla\znabla q|_{\zg})
\label{formula_ricci}
\,.
\end{eqnarray}
Next, using
$$
 \hat C^k_{ij}:= \hat \Gamma^k_{ij} - \Gamma^k_{ij} =  \nabla_{(i}q_{j)}{}^k -\frac{1}{2}\nabla^kq_{ij}   -q^{kl}\nabla_{(i}q_{j)l}+ \frac{1}{2}q^{kl}\nabla_lq_{ij} +  O(|q|_{g}^2) + \error
 \,,
$$
we can write
\begin{eqnarray}
R[g+ q] &=&
\hat g^{ij}(R_{ij}[ g] + \nabla_{k} \hat C^k_{ij}- \nabla_{i} \hat C^k_{kj} +\hat C^k_{ij}\hat C^l_{kl} - \hat C^k_{il} \hat C^l_{jk})
\\
&=&
R[ g]- g^{ik}g^{jl} q_{kl}R_{ij}[ g] + g^{ij}\nabla_{k} \hat C^k_{ij}- g^{ij} \nabla_{i} \hat C^k_{kj}
\nonumber
\\
&&
-g^{im}g^{jn} q_{mn} \nabla_{k} \hat C^k_{ij} + g^{im}g^{jn} q_{mn} \nabla_{i} \hat C^k_{kj}
\nonumber
\\
&&
+ O(|q|_g^2)+ O(|\nabla q|_g^2)
+ O(|q|_{g}^2|\nabla\nabla q|_{g})
\\
&=&
R[ g]- g^{ik}g^{jl} q_{kl} R_{ij}[ g] +g^{ik}g^{jl}   \nabla_{i} \nabla_{j}q_{kl} -\Delta \tr_g q
\nonumber
\\
&&  - 2g^{im}g^{jn} g^{kl}q_{mn}   \nabla_{i}\nabla_{k}q_{jl}
+g^{ik}g^{jl}q_{ij}\nabla_{k} \nabla_l\tr_g q
 + g^{ik}g^{jl}q_{ij}\Delta q_{kl}
\nonumber
\\
&&
+ O(|q|_g^2)+ O(|\nabla q|_g^2) + O(|q|_{g}^2|\nabla\nabla q|_{g})
\\
&=&
R[ g]-  q^{ij} R_{ij}[ {\zg} ]
 +  g^{ik}g^{jl}\znabla_{k} \znabla_{l}q_{ij} - g^{kl}g^{mn}\znabla_m\znabla_nq_{kl}
\nonumber
\\
&&
 + g^{ik}g^{jl}  g^{mn}(q_{ij}\znabla_{k} \znabla_l q_{mn} +q_{ij}\znabla_m\znabla_n q_{kl}
- 2q_{kl}  \znabla_{i}\znabla_{\ellm }q_{jn})
\nonumber
\\
&&
+ O(|q|_{\zg}^2)+ O(|\nabla q|_{\zg}^2) + O(|q|_{\zg}^2|\znabla\znabla q|_{\zg})
+ O(|q|_{\zg}|\znabla \lambda|_{\zg})
\nonumber
\\
&&
  + O(|\znabla q|_{\zg}|\znabla \lambda|_{\zg})+ O(|q|_{\zg}|\znabla\znabla \lambda|_{\zg})
+ O(|\lambda |_{\zg}| q|_{\zg})
\label{gql_exp}
\,,
\end{eqnarray}
 where the indices on $q_{ij}$ have been raised with the metric $\zg$.

We need a more detailed version of the above in the case $q_{xA}=q_{AB}= \lambda_{xx}=\lambda_{xA}=0$, and when the metric $\zg$ satisfies
\eq{17VIII15.1}. We start by noting that
\begin{equation}
 R[{\zg}]_{ij} = -(n-1){\zg}_{ij}\,, \quad R[{\zg}] = -n(n-1)
\,.
\label{ricci_zg}
\end{equation}
For the Christoffel symbols of ${\zg}$ we find
\begin{eqnarray}
& \mathring \Gamma^C_{AB} \,=\,  \Gamma[{\ringh}{}]^C_{AB}
\,,
\quad
\mathring\Gamma^x_{AB}\,=\, x^{-1}\big(1-\frac{k^2}{16}x^4\big){\ringh}{}_{AB}
\,,
\quad
\mathring\Gamma^x_{xA} \,=\, 0
\,,&
\label{christ_zg1}
\\
&\mathring\Gamma^x_{xx} \,=\, -x^{-1}\,, \quad
\mathring\Gamma^C_{xA}\,=\, -x^{-1}\frac{1+\frac{k}{4}x^2}{1-\frac{k}{4}x^2}\delta_A^C
\,,
\quad
\mathring\Gamma^C_{xx} \,=\, 0
\,.&
\label{christ_zg2}
\end{eqnarray}
As such, it holds that
\begin{eqnarray*}
&& g^{ik}g^{jl}\znabla_{k} \znabla_{l}q_{ij} - g^{kl}g^{mn}\znabla_m\znabla_nq_{kl}
\\
&&
 + g^{ik}g^{jl}  g^{mn}q_{ij}\znabla_{k} \znabla_l q_{mn} +g^{ik}g^{jl} q_{ij}g^{mn}\znabla_m\znabla_n q_{kl}
- 2g^{kl}q^{ij}   \znabla_{i}\znabla_{k}q_{jl}
\\
&=&
-(n-1)x^2(  x\partial_x   -  n + 3 )q_{xx} +  O(x^4)q_{xx}  +O(x^5)\partial_xq_{xx}
\\
&&
- x^4 \bzmcD _A\bzmcD ^A q_{xx}+O(x^6) \bzmcD _A\bzmcD ^A q_{xx}
+2x^6q_{xx}\bzmcD _A\bzmcD ^A q_{xx}
+ O(x^8)q_{xx}\bzmcD _A\bzmcD ^A q_{xx}
\\
&&
+x^2(1-2x^2q_{xx}) (\lambda^{AB} + O(|\lambda|_{\zg}^2))\bzmcD _A\bzmcD _B q_{xx}
\\
&&
+ O(|q|_{\zg}^2)+ O(|q|_{\zg}|\nabla q|_{\zg})  + O(|\lambda |_{\zg}| q|_{\zg})   + O(|\znabla\lambda |_{\zg}| q|_{\zg})
 \,,
\end{eqnarray*}
where $\bzmcD $ denotes the covariant derivative associated to the Riemannian metric ${\ringh}{}$, and where $\bzmcD ^A ={\ringh}{}^{AB}\bzmcD _B$.
Using the formula  \eq{Ric_formula} for $R_{ij}[\zg]$  we obtain
\begin{eqnarray}
 \nn
  \lefteqn{
R[g+ q]
 =
R[ g]  -(n-1)x^2(  x\partial_x   -  n + 2 )q_{xx} +  O(x^4)q_{xx}  +O(x^5)\partial_xq_{xx}
}
 &&
\\
 \nn
&&
- x^4 \bzmcD _A\bzmcD ^A q_{xx}+O(x^6) \bzmcD _A\bzmcD ^A q_{xx}
+2x^6q_{xx}\bzmcD _A\bzmcD ^A q_{xx}
\\
 \nn
&&
+x^2(1-2x^2q_{xx}) (\lambda^{AB} + O(|\lambda|_{\zg}^2))\bzmcD _A\bzmcD _B q_{xx}
+ O(x^8)q_{xx}\bzmcD _A\bzmcD ^A q_{xx}
\\
 \nn
    &&
    + O(|q|_{\zg}^2)
  + O(|\nabla q|_{\zg}^2)
+ O(|q|_{\zg}^2|\znabla\znabla q|_{\zg})
\\
    &&
    + O(| q|_{\zg}|\lambda |_{\zg})  + O(|\znabla q|_{\zg}|\znabla \lambda|_{\zg})+ O(|q|_{\zg}|\znabla\znabla \lambda|_{\zg})
 \,.
  \label{6IX15.1}
\end{eqnarray}

We  also need to compare $R[\zg+q]$ with $R[\zg]$.
Equation \eq{gql_exp} with $\lambda=0$ yields
\begin{eqnarray}
 R[g]
&=&R[\zg]
 +\znabla_{k}\znabla_{l}q^{kl} - \mathring \Delta\mathrm{tr}q
-q^{ij}R_{ij}[{\zg}]
\nonumber
\\
&&+ O(|q|_{\zg}^2)
+ O(|\znabla q|_{\zg}^2)
+ O(|q|_{{\zg}} |\znabla  \znabla q|_{{\zg}})
\,.
\label{exp_Ricci}
\end{eqnarray}
Again, all indices raised and lowered with ${\zg}$.
In the calculations that follow, the following formulae are useful:
\bean
 \znabla_i  \znabla_j q^{ij}
&  =& \frac 1 {\sqrt {\det \zg}} \partial_i ({\sqrt {\det \zg}}\, \znabla_j q^{ij} )
\\
 &
  =
   &
   \frac 1 {\sqrt {\det \zg}} \partial_i \big(\partial_j ({\sqrt {\det \zg}} \,  q^{ij} )
  +  {\sqrt {\det \zg}} \, \zGamma^i_{k j}  q^{kj} \big)
  \,,
\\
 \Delta_{\zg} \tr_{\zg} q
 & = &
  \frac 1 {\sqrt {\det \zg}} \partial_i \big({\sqrt {\det \zg}}\, \zg^{ij} \partial_j (\tr_{\zg} q )
   \big)
   \,.
\eeal{12VIII15.4}
Assume again that \eq{17VIII15.1} and thus \eq{ricci_zg}-\eq{christ_zg2} hold.
An application of \eq{exp_Ricci}  then gives
\begin{eqnarray}
\nn
 \lefteqn{
 R[ {\zg}+ q]
=
 R[ {\zg}]
 -\big((x \partial _{x})^2   - (n-1) \big) \mathrm{tr}_{\zg}q
+  n x \frac{1+\frac{k}{4}x^2}{ 1-\frac{k}{4}x^2}\partial_x (\mathrm{tr}_{\zg}q)
}
 &&
\\
 \nn
&&
-(n-3)   \Big(\frac{1+\frac{k}{4}x^2}{1-\frac{k}{4}x^2}\Big)^2 \tr_{\zg} q
-2\frac{1+\frac{k^2}{16}x^4}{(1-\frac{k}{4}x^2)^{2}}\tr_{\zg} q
 - \frac{x^2 }{(1-\frac{k}{4}x^2)^{2}}\bzmcD ^{A}\bzmcD _{A} (\mathrm{tr}_{\zg} q)
\\
 \nn
&&+x^2(x^2\partial^2_{xx} + 5x \partial_x    +4 )  q_{xx}
-\frac{x^2}{(1-\frac{k}{4}x^2)^{2}}\Big(
(2n-1)x(1-\frac{k^2}{16}x^4)\partial_x  q_{xx}
\\
 \nn
&&
+2(n-1)(1-\frac{3}{16}k^2x^4)q_{xx}
-n(n-3)(1+\frac{k}{4}x^2)^2q_{xx}
- \frac{k^2}{4}x^4q_{xx}
\Big)
\\
 \nn
&&
+
\frac{2x^3}{ (1-\frac{k}{4}x^2)^{2}} (x\partial_{x}+ 1)(\bzmcD ^A q_{xA})
 -2(n-2)x^3\frac{1+\frac{k}{4}x^2}{(1-\frac{k}{4}x^2)^3} \bzmcD ^Aq_{xA}
\\
&&
+\frac{x^4}{(1-\frac{k}{4}x^2)^{4}}\bzmcD ^A\bzmcD ^Bq_{AB}
+ O(|q|_{\zg}^2)
+ O(|\znabla q|_{\zg}^2)
+ O(|q|_{{\zg}} |\znabla  \znabla q|_{{\zg}})
\,.
\eeal{12VIII15.1}
Making explicit the dominant terms only, this becomes
\begin{eqnarray}
 R[ {\zg}+ q]
 \nn
&=&
 R[ {\zg}]
 + x^2\Big((1-n)x \partial_x  +(n-1)(n-2) - x^2\bzmcD ^{A}\bzmcD _{A}
 \Big)q_{xx}
\\
 \nn
&&
 -\big( (x \partial_{x})^2- nx \partial_x  +x^2\bzmcD ^{A}\bzmcD _{A} \big)(x^2 \mathrm{tr}_{{\ringh}{}}q)
\\
 \nn
&&
+2x^3(x\partial_{x} + 3-n  )\bzmcD ^A q_{xA}
+x^4 \bzmcD ^A\bzmcD ^Bq_{AB}
\\
 \nn
&&
+  O(x^4)q_{xx} + O(x^5)\partial_x q_{xx}
 + O(x^6)\bzmcD ^{A}\bzmcD _{A}q_{xx}
+ O(x^4) \mathrm{tr}_{{\ringh}{}}q
\\
 \nn
&&
+  O(x^5) \partial_x\mathrm{tr}_{{\ringh}{}}q
+  O(x^6) \partial^2_{xx}\mathrm{tr}_{{\ringh}{}}q
+ O(x^6) \bzmcD ^{A}\bzmcD _{A} \mathrm{tr}_{{\ringh}{}}q
\\
 \nn
&&
 + O(x^5)  \bzmcD ^Aq_{xA} + O(x^6)\partial_x( \bzmcD ^Aq_{xA})
 + O(x^6) \bzmcD ^A\bzmcD ^Bq_{AB}
\\
&&
+ O(|q|_{\zg}^2)
+ O(|\znabla q|_{\zg}^2)
+ O(|q|_{{\zg}} |\znabla  \znabla q|_{{\zg}})
\,.
\eeal{26VIII15.1asf}

\section{Proof of Lemma~\protect\ref{Lem:UChangeVar}}
 \label{A28XII17.1}

 In the new coordinate system $(\bar x, x^A)$, we will use $\bar \lambda$ to denote the difference between $g$ and the new reference metric
$$\mathring{\bar g} = \bar x^{-2}\Big[\mathrm{d}\bar x^2 + \Big(1 - \frac{k}{4}\bar x^2\Big)^2\ringh
 \Big]
 \,.
$$
We will accordingly use a bar to refer to the metric components of $\bar \lambda$, its Newton tensor etc. For example, we have $\bar \lambda = \bar \lambda_{xx}\, \mathrm{d}\bar x^2+ \bar \lambda_{xA} \,\mathrm{d}\bar x\mathrm{d}x^A + \bar \lambda_{AB}\,\mathrm{d}x^A\mathrm{d}x^B$.

We compute
%
\begin{align*}
x^{-2}\Big(1 - \frac{k}{4}x^2\Big)^2
	&= {\bar x}^{-2}\Big(1 - \frac{k}{4}{\bar x}^2\Big)^2 + 2{\bar x}^{n-4}\,\psi + O({\bar x}^{\min(n,2n-6)})
		\,,\\
\mathrm{d}x
	&= [1 - (n-1)\bar x^{n-2}\psi + O(\bar x^{2n-4})]\mathrm{d}{\bar x}\\
		&\qquad\qquad - [\bar x^{n-1}\bzmcD _A  \psi + O(\bar x^{2n-3})]\,\mathrm{d}{  x}^A
	\,.
\end{align*}
This implies that
\begin{align*}
x^{-2}\,\mathrm{d}x^2
	&= \bar x^{-2}\,\mathrm{d}\bar x^2  - [2(n-2){\bar x}^{n-4}\,\psi + O({\bar x}^{2n-6})]\,\mathrm{d}{\bar x}^2\\
		&\qquad\qquad  - 2[\bar x^{n-3} \bzmcD _A \psi + O({\bar x}^{2n-5})]\mathrm{d}{\bar x}\,\mathrm{d}{  x}^A + O({\bar x}^{2n-4})\mathrm{d}{  x}^A\,\mathrm{d}{  x}^B
			\,,
\end{align*}
and
\begin{align*}
x^{-2}\Big(1 - \frac{k}{4}x^2\Big)^2\,\ringh_{AB}(x^C)\,\mathrm{d}x^A\,\mathrm{d}x^B
	&= \bar x^{-2}\Big(1 - \frac{k}{4}\bar x^2\Big)^2\,\ringh_{AB}\,\mathrm{d}x^A\,\mathrm{d}x^B \\
		&
+ 2[\bar x^{n-4}\psi + O(\bar x^{\min(n,2n-6)})]\,\ringh_{AB}\,\mathrm{d}x^A\,\mathrm{d}x^B
		\,.
\end{align*}
It follows that,
\begin{align*}
\bar \lambda 	
	&= [\lambda_{xx}(\bar x, x^C) - 2(n-2)\bar x^{n-4}\psi(x^C) + O(\bar x^{2n-6})]\, \mathrm{d}\bar x^2\\
		&\qquad  + \red{2} [\lambda_{xA}(\bar x, x^C) - \bar x^{n-3} \bzmcD _A \psi(x^C)  + O(\bar x^{2n-5})]  \,\mathrm{d}\bar x\mathrm{d}x^A\nonumber\\
		&\qquad + [\lambda_{AB}(\bar x,x^C) + 2\bar x^{n-4}\psi(x^C)\,\ringh_{AB}(x^C)
  + O(\bar x^{\min(n,2n-6)})]\,\mathrm{d}x^A\mathrm{d}x^B
			\,.
\end{align*}
%

We now proceed to compute the tensor $T$. We note that
\[
 \bar x^2 \Big(1 - \frac{k}{4}\bar x^2\Big)^{-2} =  x^2 \Big(1 - \frac{k}{4} x^2\Big)^{-2} + 2\bar x^n\,\psi + O(\bar x^{\min(2n-2,n+2)})
 	\,.
\]
This leads to, in view of \eqref{30VII17.1},
\begin{align*}
\tr_{\mathring{\bar g}}(\bar \lambda)
	& = \bar x^2 \,\bar \lambda_{xx} + \bar x^2 \Big(1 - \frac{k}{4}\bar x^2\Big)^{-2} \,\ringh^{AB}\,\bar \lambda_{AB}\\
	& = \bar x^2 \,\lambda_{xx} + \bar x^2 \Big(1 - \frac{k}{4}\bar x^2\Big)^{-2} \,\ringh^{AB}\,\lambda_{AB}\\
		&\qquad\qquad - 2(n-2)\bar x^{n-2}\,\psi + 2(n-1)\bar x^{n-2}\Big(1 - \frac{k}{4}\bar x^2\Big)^{-2}\,\psi
			+ O(\bar x^{\min(n+2,2n-4)})\\
	& = \tr_{\mathring{g}}(\lambda)
		 + 2\bar x^{n-2}\,\psi + k(n-1)\bar x^{n}\,\psi
			+ O(\bar x^{\min(n+2,2n-4)})
				\,,\\
\bar T^{xx}
	&= \bar x^4 \,\bar \lambda_{xx} - \tr_{\mathring{\bar g}}(\bar \lambda)\,\bar x^2\\
	&= T^{xx} - 2(n-1)\bar x^{n}\,\psi - k(n-1)\bar x^{n + 2}\,\psi
			+ O(\bar x^{\min(n+4,2n-2)})
				\,,
\\
\bar T^{xA}
	&= \bar x^4\Big(1 - \frac{k}{4}\bar x^2\Big)^{-2} \ringh^{AB} \bar \lambda_{xB}\\
	&= \bar x^4\Big(1 - \frac{k}{4}\bar x^2\Big)^{-2}
  \left(\ringh^{AB} \lambda_{xB} - \bar x^{n-3}\,\bzmcD ^A \psi \right)
   + O(\bar x^{2n-1})
\\
	&= \Big[ x^4\Big(1 - \frac{k}{4} x^2\Big)^{-2} +  O(\bar x^{n+2})\Big]\ringh^{AB}\, \lambda_{xB}  - \bar x^{n+1}\,\bzmcD ^A \psi + O(\bar x^{2n-1})
  + O(\bar x^{n+3})
\\
	&= T^{xA} - \bar x^{n+1}\,\bzmcD ^A \psi
            +  O(\bar x^{\min{}(n+3,2n-1)})
		\,,\\
\bar \lambda^{AB}
	&= \bar x^4\Big(1 - \frac{k}{4}\bar x^2\Big)^{-4}\,\ringh^{AC}\,\ringh^{BD} \bar \lambda_{CD}\\
	&= \bar x^4\Big(1 - \frac{k}{4}\bar x^2\Big)^{-4}\,\ringh^{AC}\,\ringh^{BD} \lambda_{CD} + 2\bar x^n\Big(1 - \frac{k}{4}\bar x^2\Big)^{-4}\,\psi\,\ringh^{AB} + O(\bar x^{\min(n+4,2n-2)})\\
	&= \lambda^{AB}  + 2\bar x^n\,\psi\,\ringh^{AB} + 2k\,\bar x^{n+2}\,\psi\,\ringh^{AB}  + O(\bar x^{\min(n+4,2n-2)})
		\,,\\
\bar T^{AB}
	&= \bar \lambda^{AB} - \bar x^2\Big(1 - \frac{k}{4}\bar x^2\Big)^{-2}\tr_{\mathring{\bar g}}(\bar \lambda)\,\ringh^{AB}\\
	&= \bar \lambda^{AB} - \bar x^2\Big(1 - \frac{k}{4}\bar x^2\Big)^{-2}\tr_{\mathring{g}}(\lambda)\,\ringh^{AB}
		-2 \bar x^n \psi\,\ringh^{AB} - kn\,\bar x^{n+2}\,\psi\,\ringh^{AB} + O(x^{\min(n+4,2n-2)})\\
	&= T^{AB} - k(n-2) \bar x^{n+2}\,\psi\,\ringh^{AB}+ O(\bar x^{\min(n+4,2n-2)})
		\,.
\end{align*}
This completes the proof. 

\section{A convenient cut-off function}
 \label{App27XII17.1}

In this appendix, we construct, for small $\epsilon > 0$, a cut-off function $\varphi_\epsilon \in C^\infty(\RR)$ such that $\varphi_\epsilon \equiv 1$ in $(-\infty, \epsilon^2)$, $\varphi_\epsilon \equiv 0$ in $(\epsilon,\infty)$, and
\[
|\varphi_\epsilon(x)| \leq C \text{ and } |\varphi_\epsilon'(x)| \leq \frac{C}{x|\ln x|}
\]
for some constant $C$ independent of $\epsilon$,
 together with
\bel{29XII17.2}
0 = \int_0^\infty (n-2) \varphi_\epsilon(x)\,x^{n-3}\,dx = - \int_0^\infty \varphi_\epsilon'(x)\,x^{n-2}\,dx
	\,.
\ee

Let $\chi, \zeta \in C^\infty(\RR)$ such that $\chi \equiv 1- \zeta \equiv 1$ in $(-\infty,0)$, $\chi \equiv \zeta \equiv 0$ in $(1,\infty)$, $\zeta \geq 0$ in $(0,1)$, $\zeta = 1$ in $(1/2,3/4)$ and $\chi \equiv \zeta$ in $(1/2,\infty)$. See Figure \ref{F29XII17.1}.

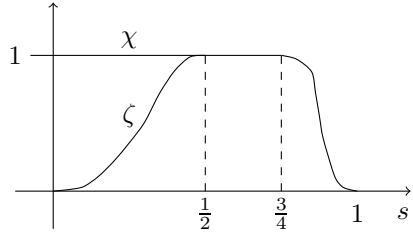
\begin{figure}[h]
\begin{center}
\begin{tikzpicture}
\draw[->] (-0.5,0) -- (4.7,0);
\draw[->] (0,-0.5)--(0,2.5);
\draw(-0.3,1.8)--(3,1.8);
\draw plot[smooth] coordinates{(0,0)(0.2,0.02)(0.4,0.06)(0.6,0.2) (0.8,0.4) (1,0.64) (1.2, 0.91) (1.4,1.3) (1.6, 1.6) (1.8,1.78) (2,1.8)};
\draw plot[smooth] coordinates{(3,1.8)(3.2,1.75) (3.4,1.58) (3.45,1.3) (3.5,1)  (3.55,0.7) (3.7,0.2) (3.8,0.06)(3.9,0.02) (4,0)};
\draw[dashed](3,0)--(3,1.8);

\draw[dashed](2,0)--(2,1.8);

\draw (4.6,-0.3) node {$s$}
(4,-0.3) node {$1$}
(2,-0.3) node {$\frac{1}{2}$}
(3,-0.3) node {$\frac{3}{4}$}
(-0.5,1.8) node {$1$}
(1,2) node {$\chi$}
(1,1) node {$\zeta$}
;
\end{tikzpicture}
\end{center}
\caption{The functions $\chi$ and $\zeta$.}
\label{F29XII17.1}
\end{figure}%

 For small $\epsilon > 0$, define
\[
\varphi_\epsilon(x) = \chi\Big(2 + \frac{\ln x}{|\ln \epsilon|}\Big) - a_\epsilon\,\zeta\Big(2 + \frac{\ln x}{|\ln \epsilon|}\Big)
	\,,
\]
where $a_\epsilon$ is a constant which is chosen so that \eqref{29XII17.2} holds, i.e. $a_\epsilon = b_\epsilon\,c_\epsilon^{-1}$ where
\begin{align*}
b_\epsilon
	&=  \int_{0}^\epsilon \chi\Big(2 + \frac{\ln x}{|\ln \epsilon|}\Big)\,x^{n-3}\,dx
		\,,\\
c_\epsilon
	&=\int_{0}^\epsilon \zeta\Big(2 + \frac{\ln x}{|\ln \epsilon|}\Big)\,x^{n-3}\,dx
	\,.
\end{align*}
By construction we have
\begin{equation}\label{27XII17.1}
  \varphi_\epsilon=\left\{
     \begin{array}{ll}
       1, & \hbox{$0<x<\epsilon^2$;} \\
       0, & \hbox{$x>\epsilon$.}
     \end{array}
   \right.
\end{equation}
Note that
\begin{align}
 \label{27XII17.3}
\Big|\int_0^{\epsilon^{3/2}} \chi\Big(2 + \frac{\ln x}{|\ln \epsilon|}\Big)\,x^{n-3}\,dx \Big|
	&\leq \frac{1}{n-2}\sup_\RR |\chi|\,\epsilon^{\frac{3}{2}(n-2)}
	\,,\\
\Big|\int_0^{\epsilon^{3/2}} \zeta\Big(2 + \frac{\ln x}{|\ln \epsilon|}\Big)\,x^{n-3}\,dx \Big|
	&\leq \frac{1}{n-2}\sup_\RR |\zeta|\,\epsilon^{\frac{3}{2}(n-2)}
	\,,
\end{align}
and
\begin{align*}
\int_{\epsilon^{3/2}}^\epsilon \chi\Big(2 + \frac{\ln x}{|\ln \epsilon|}\Big)\,x^{n-3}\,dx
	&= \int_{\epsilon^{3/2}}^\epsilon \zeta\Big(2 + \frac{\ln x}{|\ln \epsilon|}\Big)\,x^{n-3}\,dx\\
	&\geq \int_{\epsilon^{3/2}}^{\epsilon^{5/4}} x^{n-3}\,dx
		= \frac{1}{n-2}(\epsilon^{\frac{5}{4}(n-2)} - \epsilon^{\frac{3}{2}(n-2)})
			\,.
\end{align*}
The above implies that
\begin{equation}\label{27XII17.2}
\lim_{\epsilon \rightarrow 0} a_\epsilon = \lim_{\epsilon \rightarrow 0} \frac{b_\epsilon}{c_\epsilon} = 1
	\,.
\end{equation}
It is readily seen that $\varphi_\epsilon$ satisfies all the needed requirements.

\bigskip

\noindent{\sc Acknowledgements:}   PTC was supported in part by  the Austrian Science Fund (FWF) under project  {P29517-N27  and by the Polish National Center of Science (NCN) 2016/21/B/ST1/00940; he acknowledges the friendly hospitality of the IHES, Bures-sur-Yvette, during part of work on this paper}. GG's research was supported by  NSF grants DMS-1313724 and DMS-1710808.
TTP acknowledges financial support by the Austrian Science Fund (FWF) under the project P 28495-N27.
The authors are grateful to the Erwin Schr\"odinger Institute, Vienna, for hospitality and support during part of work on this paper.

\bibliographystyle{amsplain}
\bibliography{%
../references/reffile,%
../references/newbiblio,%
../references/hip_bib,%
../references/newbiblio2,%
../references/bibl,%
../references/howard,%
../references/bartnik,%
../references/myGR,%
../references/newbib,%
../references/Energy,%
../references/netbiblio,%
../references/PDE,%
ChruscielGallowayNguyenPaetz-minimal,%
aspect}

\end{document}